\documentclass[12pt]{article}
\usepackage{amsmath, amsthm}
\usepackage{psfrag}
\usepackage{enumerate}
\usepackage{hyperref}
\hypersetup{
    colorlinks  = true,
    linkcolor   = magenta,
    citecolor   = blue
}
\usepackage{natbib}
\usepackage{amsfonts}
\usepackage[demo]{graphicx}
\usepackage[font=small, justification=centering]{caption}
\usepackage[utf8]{inputenc} 
\usepackage[T1]{fontenc}
\usepackage{xcolor}
\usepackage{cmbright}
\usepackage[hmarginratio=1:1, vmarginratio=1:1, marginparwidth=4cm, textheight=23cm, textwidth=17cm]{geometry}
\usepackage{algorithm}
\usepackage[noend]{algpseudocode}

\newtheorem{lemma}{Lemma}
\newtheorem{proposition}{Proposition}
\newtheorem{corollary}{Corollary}
\newtheorem{remark}{Remark}

\def\E{\mathbb{E}}
\def\T{\mathcal{T}}
\def\bern{\mathrm{Ber}}
\def\pois{\mathrm{Pois}}
\def\cusum{\mathrm{CUSUM}}
\def\lr{\mathrm{LR}}
\def\pl{\mathrm{PL}}

\def\iid{\protect\overset{\mathrm{i.i.d.}}{\sim}}
\newcommand{\ind}[1]{\mathbb{I}_{\{#1\}}} 
\def\spacingset#1{\renewcommand{\baselinestretch}{#1}\small\normalsize} \spacingset{1}

\title{\bf Exact Tests for Offline Changepoint Detection in Multichannel Binary and Count Data with Application to Networks}

\author{
    \normalsize{\textbf{Shyamal K. De}}\thanks{De's research is supported in part by the Science and Engineering Research Board grant MTR/2017/000503 under the MATRICS scheme.}\\
    \normalsize{School of Mathematical Sciences}\\
    \normalsize{NISER, HBNI}\\
    \normalsize{OD 752050, India}\\
    \normalsize{\texttt{shyamalkd@gmail.com}}\\
    \and
    \normalsize{\textbf{Soumendu Sundar Mukherjee}}\thanks{Mukherjee's research is supported in by an INSPIRE Faculty Fellowship from the Department of Science and Technology, Government of India.}\\
    \normalsize{Interdisciplinary Statistical Research Unit}\\ 
    \normalsize{Indian Statistical Institute, Kolkata}\\
    \normalsize{WB 700108, India}\\
    \normalsize{\texttt{soumendu041@gmail.com}}
}

\usepackage{fancyhdr}
\date{}

\begin{document}

\maketitle

\begin{abstract}
We consider offline detection of a single changepoint in binary and count time-series. We compare exact tests based on the cumulative sum (CUSUM) and the likelihood ratio (LR) statistics, and a new proposal that combines exact two-sample conditional tests with multiplicity correction, against standard asymptotic tests based on the Brownian bridge approximation to the CUSUM statistic. We see empirically that the exact tests are much more powerful in situations where normal approximations driving asymptotic tests are not trustworthy: (i) small sample settings; (ii) sparse parametric settings; (iii) time-series with changepoint near the boundary.

We also consider a multichannel version of the problem, where channels can have different changepoints. Controlling the False Discovery Rate (FDR), we simultaneously detect changes in multiple channels. This ``local'' approach is shown to be more advantageous than multivariate global testing approaches when the number of channels with changepoints is much smaller than the total number of channels.
    
As a natural application, we consider network-valued time-series and use our approach with (a) edges as binary channels and (b) node-degrees or other local subgraph statistics as count channels. The local testing approach is seen to be much more informative than global network changepoint algorithms.
\end{abstract}

\noindent%

\spacingset{1.45} 
\section{Introduction}
\label{sec:intro}
Changepoint analysis is an important problem in statistics with roots in statistical quality control \citep{page1954continuous,page1957problems,girshick1952bayes}. The goal of changepoint analysis is to decide if there are distributional changes in a given time-series (the detection part), and estimate the change if any (the estimation part). There is a huge body of literature on the univariate changepoint problem. An excellent treatment can be found in the book \cite{brodsky2013nonparametric}. 

Some notable works on the multivariate version of the problem are \cite{zhang2010detecting,siegmund2011detecting,srivastava1986likelihood,james1992asymptotic} in parametric settings, and \cite{harchaoui2009kernel,lung2011homogeneity,chen2015} in non-parametric settings.

We should mention that there are two types of changepoint problems: (a) offline, where the whole time-series is available to the statistician; (b) online, where data is still arriving at the time of analysis. We will be concerned with the offline problem in this article.

Although a lot of work has been done on changepoint detection for continuous time-series data, results for discrete data are lacking, especially in ``small sample'' settings where the length of the time-series is relatively small. In this article, our main goal is to develop methods for offline changepoint detection for binary and count data that have good performance in small sample settings.

We adapt well-known conditional two-sample tests for binary and count data to the changepoint setup using a multiple testing approach. We also consider exact tests based on natural statistics such as the CUSUM statistic and the LR statistic. We conduct a comprehensive small-sample power analysis of these tests and compare them against the large sample CUSUM test based on a Brownian bridge approximation \citep{brodsky2013nonparametric}. We find that, in small sample scenarios, and in cases where the true changepoint lies near the boundary, the exact tests are significantly more powerful than large sample tests.

Although these methods are developed for single changepoint problems, they seem to work well when multiple changepoints are present, especially if there is one strong change. We report some empirical findings in this direction in the appendix.

We then consider multichannel binary or count time-series. Using a False Discovery Rate (FDR) controlling mechanism, we simultaneously test for changepoints in all the channels. This ``local'' approach vastly outperforms the ``global'' approach of using some statistic of all channels together (e.g., a vector CUSUM statistic), when the number of channels with changepoints is much smaller than the total number of channels.

As an application of this approach, we consider local vs. global testing in network-valued time-series. Although there has been a recent surge of interest in network changepoints \citep{peel2015detecting,roy2017change,mukherjee2018thesis,wang2018optimal,padilla2019change,zhao2019change,bhattacharjee2020change,bhattacharyya2020optimal}, we note that the existing works are focused on large sample asymptotics and use global statistics for detection or estimation. If we use edges (resp. node degrees or some other local subgraph statistics such as local triangle counts) as separate channels, then we have multichannel binary (resp. count) data. We compare the proposed local approach against a standard CUSUM-based global approach in real-world networks. We see that, in addition to picking up strong global changes, the local approach can identify relatively weak and rare changes.

The rest of the paper is organised as follows. In Section~\ref{sec:setup}, we describe the problem set-up precisely and detail our methodology. In Section~\ref{sec:mult}, we discuss our multiple testing based local approach for multichannel changepoint detection. In Section~\ref{sec:simu}, we report our simulations: In Section~\ref{sec:exact-asymp}, we perform a comprehensive power analysis of the various proposed methods against existing approaches. In Section~\ref{sec:glob-loc}, we compare the local testing approach vs. multivariate CUSUM-based global testing approaches in multichannel problems. More detailed results are provided in the appendix. Then, in Section~\ref{sec:realdata}, we apply our methodology on two real-life examples: a time-series of US senate voting pattern networks, and another time-series of phone-call networks. We conclude the paper with a discussion in Section~\ref{sec:discuss}.

\section{Set-up and methodology}
\label{sec:setup}
Suppose that we have time-indexed independent variables $X_1, \ldots, X_T$, with $\E(X_i) = \pi_i$. We want to test if the $\pi_i$ have changed over time.
The single changepoint testing problem is:
\begin{align}\label{eq:cpd-setup}\nonumber
  H_0: \,\,&\pi_i = \pi \text{ for all }1 \le i \le T\,\, \text{(no change)} \\ \nonumber
   &\text{vs. } \\
  H_1: \,\,&\exists 1 \le \tau \le T - 1 \text{ such that } \pi_i = \pi_1 \ind{i \le \tau} + \pi_2 \ind{i > \tau}\,\, \text{(at least one change)}.
\end{align}
We are interested in the situation where the $X_i$'s are binary or counts. The binary case is obviously modelled by a independent Bernoulli time-series, whereas we model count data using the Poisson distribution. Keeping that in mind, let us now discuss some natural test statistics for the testing problem \eqref{eq:cpd-setup}. We begin by deriving the likelihood ratio statistic.
\vskip10pt
\noindent
\textbf{The Likelihood Ratio (LR) statistic.}
\vskip5pt
\noindent
\textbf{Binary data:} Note that under $H_1$, the likelihood of the data is
\[
  L(\pi_1, \pi_2, \tau) = \pi_1^{S_{\tau}} (1 - \pi_1)^{\tau - S_{\tau}} \times \pi_2^{S_T - S_{\tau}} (1 - \pi_2)^{T - \tau - (S_T - S_{\tau})}.
\]
The maximizers of $L$ for a fixed $\tau$ are $\hat{\pi}_1 = \frac{S_{\tau}}{\tau}$ and $\hat{\pi}_2 = \frac{S_T - S_{\tau}}{T - \tau}$. Thus the profile log-likelihood for $\tau$ is
\[
  \ell_{\pl}(\tau) = -\tau H(\hat{\pi}_1) - (T - \tau) H(\hat{\pi}_2),
\]
where $H(x) = - x \log x - (1 - x) \log (1 - x)$  is the entropy of a $\bern(x)$ variable. Define
\[
  \T_b = \min_{1 \le t \le T - 1} \bigg[t H\bigg(\frac{S_t}{t}\bigg) + (T - t) H\bigg(\frac{S_T - S_t}{T - t}\bigg)\bigg].
\]
Then the LR statistic is
\begin{equation}\label{eq:LR-bin}
  \T_{\lr}^{(b)} = -2(\ell_0 - \ell_1) = -2 (- TH(S_T/T) + \T_b).
\end{equation}
We would reject $H_0$ for large values of this statistic.
\vskip5pt
\noindent
\textbf{Count data:} Recall that we are modeling counts using the Poisson distribution. Under $H_1$, the likelihood of the data is
\[
  L(\pi_1, \pi_2, \tau) \propto e^{-\tau \pi_1}\pi_1^{S_{\tau}}  e^{-(T - \tau) \pi_2}\pi_2^{S_T - S_{\tau}}.
\]
The maximizers of $L$ for a fixed $\tau$ are $\hat{\pi}_1 = \frac{S_{\tau}}{\tau}$ and $\hat{\pi}_2 = \frac{S_T - S_{\tau}}{T - \tau}$. Therefore the profile log-likelihood for $\tau$ is
\[
  \ell_{\pl}(\tau) = -\tau G(\hat{\pi}_1) - (T - \tau) G(\hat{\pi}_2),
\]
where $G(x) = x (1 - \log x)$. Define
\[
  \T_c = \min_{1 \le t \le T - 1} \bigg[t G\bigg(\frac{S_t}{t}\bigg) + (T - t) G\bigg(\frac{S_T - S_t}{T - t}\bigg)\bigg].
\]
Then the LR statistic is
\begin{equation}\label{eq:LR-count}
  \T_{\lr}^{(c)} = -2(\ell_0 - \ell_1) = -2 (- TG(S_T/T) + \T_L).
\end{equation}
We would reject $H_0$ for large values of this statistic.
\vskip10pt
\noindent
\textbf{CUSUM statistic.}
\vskip5pt
\noindent
A well-known and often-used statistic in changepoint problems is the so-called CUSUM statistic. For $0< a < b <1$, suppose $aT$ and $bT$ are known upper and lower bounds on the locations of the potential changepoints. For $0\le \delta \le 1$, the CUSUM statistic is defined as
\begin{equation}\label{eq:CUSUM}
    \T_{\cusum}^{(\delta)} = \max_{ aT  \le t \le bT } \left[ \frac{t}{T}\left(1-\frac{t}{T}\right) \right]^{\delta} \Bigg | \frac{S_t}{t} - \frac{S_T - S_t}{T - t} \Bigg |.
\end{equation}
This can be used with both binary and count data.

\subsection{Asymptotic tests}
First, we will consider an asymptotic test based on the CUSUM statistic \eqref{eq:CUSUM}. The asymptotic null distribution can be calculated using a Brownian bridge approximation. For details see, e.g., \cite{brodsky2013nonparametric}.
\begin{proposition}\label{prop:BB}
Let $B^0(t)$ denote a standard Brownian bridge. Under $H_0$, $\pi_i = \pi$ for all $i$, and 
\[
    \frac{\sqrt{T} \, \T_{\cusum}^{(\delta)}}{\sqrt{\pi(1-\pi)}}  \xrightarrow[T \to \infty]{ \mathcal{L}}  M_{ab}^{(\delta)},
\]
where $M_{ab}^{(\delta)} =  \max_{a \le t \le b} \frac{|B^0(t)|}{(t(1 - t))^{1 - \delta}}$.
\end{proposition}

\begin{corollary}\label{cor:BBapprox}
Let $\widehat{\pi} = \frac{1}{T} \sum_{s=1}^T X_s \xrightarrow[H_0, T \to \infty]{a.s.} \pi$. Then, under $H_0$, 
\[
    \frac{\sqrt{T} \, \T_{\cusum}^{(\delta)}}{\sqrt{\widehat{\pi}(1-\widehat{\pi})}}  \xrightarrow[T \to \infty]{ \mathcal{L}} M_{ab}^{(\delta)}.
\]
\end{corollary}
Using this result we can perform an asymptotic test for $H_0$ in the ``large sample'' regime where $T$ is large. This test would be good when $\pi$ is not too small (so that the underlying normal approximations to the partial sums $\sum_{s = 1}^t X_s$ go through). It is well-known that in this asymptotic framework the choice $\delta = 1/2$ is the best for estimation (See, e.g., \cite{brodsky2013nonparametric}, Chapter 3), while $\delta = 1$ is the best for minimizing type-1 error, $\delta = 0$ for minimizing type-2 error. However, in the small sample situations explored in this paper we do not see such a clear-cut distinction (see Section~\ref{sec:simu}).

\subsection{Conditional tests}
\label{sec:conditionaltests}
Our exact tests are based on the following simple lemma.
\begin{lemma}\label{lem:suff}
Suppose $X_1$, \ldots, $X_T$ are independent Bin$(n_i, \pi)$ (or Poisson$(\pi)$). Let $S_i = \sum_{j = 1}^i X_i$. Then the joint distribution of $(S_1, \ldots, S_{T - 1})$ given $S_T$ does not depend on $\pi$.
\end{lemma}
\begin{proof}
Since $S_T$ is sufficient for $\pi$, the distribution of $(X_1, \ldots, X_T)$ given $S_T$ does not depend on $\pi$. Hence the same holds for $(S_1, \ldots, S_{T - 1})$.
\end{proof}
\noindent
\textbf{Approach 1.}
\vskip5pt
\noindent
By Lemma~\ref{lem:suff}, $\T_{LR} \mid S_T$ does not depend on $\pi$ under $H_0$. So we can do an exact conditional test. In fact, we can use the statistics $\T_{b}$ (or $\T_c$) which is equivalent to $\T_{\lr}^{(b)}$ for conditional testing. Similarly, we can do a CUSUM based exact test, since the CUSUM statistic $\T_{\cusum}^{(\delta)}$ is a function of the partial sums $S_t, t < T$.
\vskip10pt
\noindent
\textbf{Approach 2.}
\vskip5pt
\noindent
Note that we can decompose $H_1$ as a disjoint union of the following $(T - 1)$ hypotheses: 
\[
  H_{1i}: \tau = i, 1 \le i \le T - 1, 
\]
and test these separately against $H_0$, and, finally, rejecting $H_0$ if one of these $T - 1$ hypotheses gets rejected.

\vskip5pt
\noindent
\textbf{Binary data:}
Suppose $X_i\sim\bern(\pi_i)$. Note that if we use $S_i = \sum_{j = 1}^i X_j$ as a test statistic for testing $H_0$ against $H_{1i}$, then, under $H_0$, 
\[
  S_i \mid S_T \sim \mathrm{Hypergeometric}(i, S_T, T).
\]
Therefore, we get a $p$-value $p_i$ from this conditional distribution as
\[
  p_i = \sum_{q \,:\, f(q;\, i, S_T, T) \le f(S_i;\, i, S_T, T)} f(q;\, i, S_T, T),
\]
where $f(q;\, i, S_T, T)$ is the PMF of the Hypergeometric$(i, S_T, T)$ distribution.

\vskip5pt
\noindent
\textbf{Count data:}
For count data $X_i \sim$ Poisson$(\pi_i)$, we can use the same procedure as above using the observation that, under $H_0$,
\[
  S_i \mid S_T \sim \mathrm{Binomial}\bigg(S_T, \frac{i}{T}\bigg).
\]
In this case, we get a $p$-value $p_i$ from the above conditional distribution as
\[
  p_i = \sum_{q \,:\, g(q;\, S_T, i/T) \le g(S_i;\, S_T, i/T)} g(q;\, S_T, i/T),
\]
where $g(q;\, S_T, i/T)$ is the PMF of the Binomial$(S_T, i/T)$ distribution.

\vskip10pt
\noindent
\textbf{Multiplicity correction:}
Once we get hold of the individual $p$-values, we can try to control the \emph{familywise error rate} (FWER). It follows from Lemma~\ref{lem:suff} that $(p_1, \ldots, p_{T - 1})$ given $S_T$ does not depend on $\pi$ under $H_0$. Thus we can exactly simulate the distribution of $p_{(1)}$ using Monte Carlo. Denoting by $r_{\alpha, T}$ the lower $\alpha$-th quantile of $p_{(1)}$, we reject $H_0$, if $p_{(1)} \le r_{\alpha, T}$.

\subsection{Changepoint estimation}
\label{sec:estimation}
While we are interested in changepoint detection, the testing methods give bona-fide estimators of the underlying changepoint. For example, the likelihood ratio statistics are based on maximizing the profile log-likelihood $\ell_{\pl}(\tau)$ and the maximizer gives an estimate of $\tau$. Similarly, for the CUSUM statistic, the maximizer in the definition gives one estimate. As for the conditional testing approach, the minimizing index of the individual $p$-values gives an estimate of the changepoint. One can show that, under a single changepoint model, these estimates are consistent, because all these objective functions are based on the cumulative average $S_t/t$, and one can use the fact that a properly rescaled version of this process converges to a Brownian motion under the null hypothesis of no changepoints. For example, an analysis of the CUSUM estimator along these lines can be found in \cite{brodsky2013nonparametric}. In Section~\ref{sec:realdata}, we obtain channel-specific estimates of changepoints in this way and plot their histograms (see Figures~\ref{fig:uss_hist}, \ref{fig:mit_hist}, and \ref{fig:mitdeg_hist}).

A statistically valid procedure for simultaneous detection and estimation may be obtained by using an even-odd sample splitting: separate out the observations with even and odd time indexes, use the even ones for testing, and based on the decision, use the odd one for further estimation. However, as with any sample splitting method, this method will suffer a loss in power in small sample scenarios.

\section{The multichannel case: global vs. local testing}
\label{sec:mult}
Suppose we observe an $m$-variate $(m>1)$ independent time-series 
\[
    \boldsymbol{X}_1, \dots, \boldsymbol{X}_{\tau} \iid F_1, \, \boldsymbol{X}_{\tau+1}, \dots, \boldsymbol{X}_{T} \iid F_2,
\]
where $F_1$ and $F_2$ are $m$-variate distributions. 
We would like to test the global null $H_0:$ ``no change in the $m$-variate time-series'', i.e., $H_0: \tau = T$. Since permutations of $\boldsymbol{X}_{1}, \dots, \boldsymbol{X}_{T}$ are equally likely under $H_0$, a natural approach for testing $H_0$ is to adopt a permutation test using the global CUSUM statistic 
\[
C^{(\delta)} = \max_{1 \le t \le T-1} \left[ \frac{t}{T}\left(1-\frac{t}{T}\right) \right]^{\delta} \bigg|\bigg| \frac{1}{t} \sum_{i=1}^t \boldsymbol{X}_i - \frac{1}{T-t} \sum_{i=t+1}^T \boldsymbol{X}_i \bigg|\bigg| \, \text{ for } \delta \in [0,1],
\]
where $|| . ||$ denote any suitable norm in $\mathbb{R}^m$. Permuting the time-series multiple times, we construct a randomized size-$\alpha$ test that rejects $H_0$ for a large value of observed $C^{(\delta)}$. However, this test cannot determine which channels were responsible for the global change.   

Now suppose that $\boldsymbol{X}_i = (X_{i, 1}, \ldots, X_{i, m})$ for $i = 1, \ldots,T$, and the time-series for the $j$-th channel is
\[
    X_{j, 1}, \ldots, X_{j, \tau} \iid F_{j, 1}, \, X_{j, \tau + 1}, \ldots, X_{j, T} \iid F_{j, 2} \, \text{ for } j = 1, 2, \ldots, m.
\] 
In this article, $F_{j, 1} \equiv \bern(p_{j, 1}) \text{ or } \pois(\lambda_{j, 1})$ and $F_{j, 2} \equiv \bern(p_{j, 2}) \text{ or } \pois(\lambda_{j, 2})$ depending on whether we deal with binary or count data. The global null $H_0$ is equivalent to $\cap_{j=1}^m H_{0, j}$ where $H_{0, j}:$ ``no change in the $j$-th channel''.  A local approach for testing $H_0$ would be to compute $p$-values corresponding to $H_{0, j}$ for $j = 1, \ldots, m$, and apply some suitable multiple testing procedure controlling FWER or FDR. Note that FDR equals FWER under $\cap_{j=1}^m H_{0, j}$. Since FDR controlling methods are known to be more powerful than traditional FWER controlling methods such as Bonferroni and Holm's methods (\cite{holm1979simple}) when $m$ is large, we use some popular methods for FDR control.  

In this article, we consider the celebrated Benjamini-Hochberg (BH) step-up procedure proposed by \cite{benjamini1995controlling}. \cite{BenjaminiYekutieli01} proved that the BH method controls FDR at a pre-fixed level when $p$-values are mutually independent or they have certain positive dependence.

Let $\mathcal{R} = \{1\le j \le m: H_{0, j} \text{ is rejected} \}$ be the rejection set obtained from some FDR controlling procedure. The global null $H_0$ is rejected if and only if $\mathcal{R}$ is nonempty. This test for $H_0$ is referred as a local test $\phi := \ind{\mathcal{R} \neq \emptyset}$.

\begin{remark}
If FDR is controlled at level $\alpha$, then $\phi$ is a valid level-$\alpha$ test for the global null $H_0$ since $P_{H_0}( H_0 \text{ rejected}) = P_{H_0}(\mathcal{R} \neq \emptyset) = P_{H_0}( \cup_{j=1}^m H_{0, j} \text{ rejected})=$ FWER $=$ FDR $\le \alpha$. 
\end{remark}

\begin{remark}
    The local testing approach enjoys a few advantages over the global testing approach. First, local testing is much more informative in the sense that channels responsible for the global change, if any, are also determined. Second, under the rare signal regime where signals are available only in a few out of a large number of channels, global tests may fail to detect a change whereas local tests are more likely to detect the change as they scrutinize all channels. These points are empirically demonstrated in the simulations of Section~\ref{sec:glob-loc}.
\end{remark}

\begin{remark}
    Although we have formulated the local testing approach for a single global changepoint so as to compare it to the global testing approach, it is clear that the former applies to situations where individual channels have different changepoints. This advantage of the local testing approach over the global testing approach will be clear in Section~\ref{sec:realdata}, where we plot histograms of detected local changepoints.
\end{remark}

\section{Simulations}
\label{sec:simu}
\def\pois{\mathrm{Pois}}

\begin{figure}[!ht]
    \centering
    \begin{tabular}{cc}
        \includegraphics[scale = 0.45]{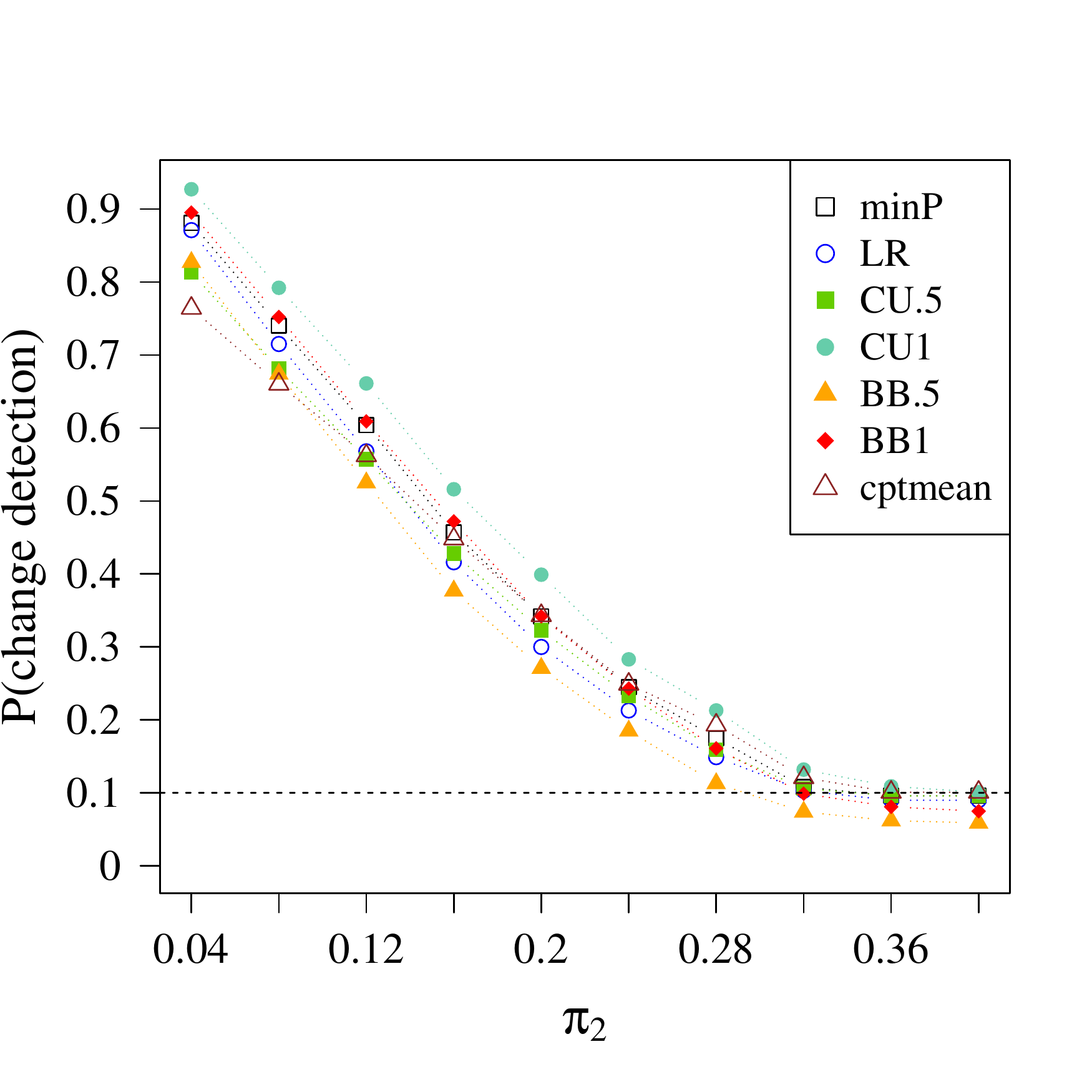} &
        \includegraphics[scale = 0.45]{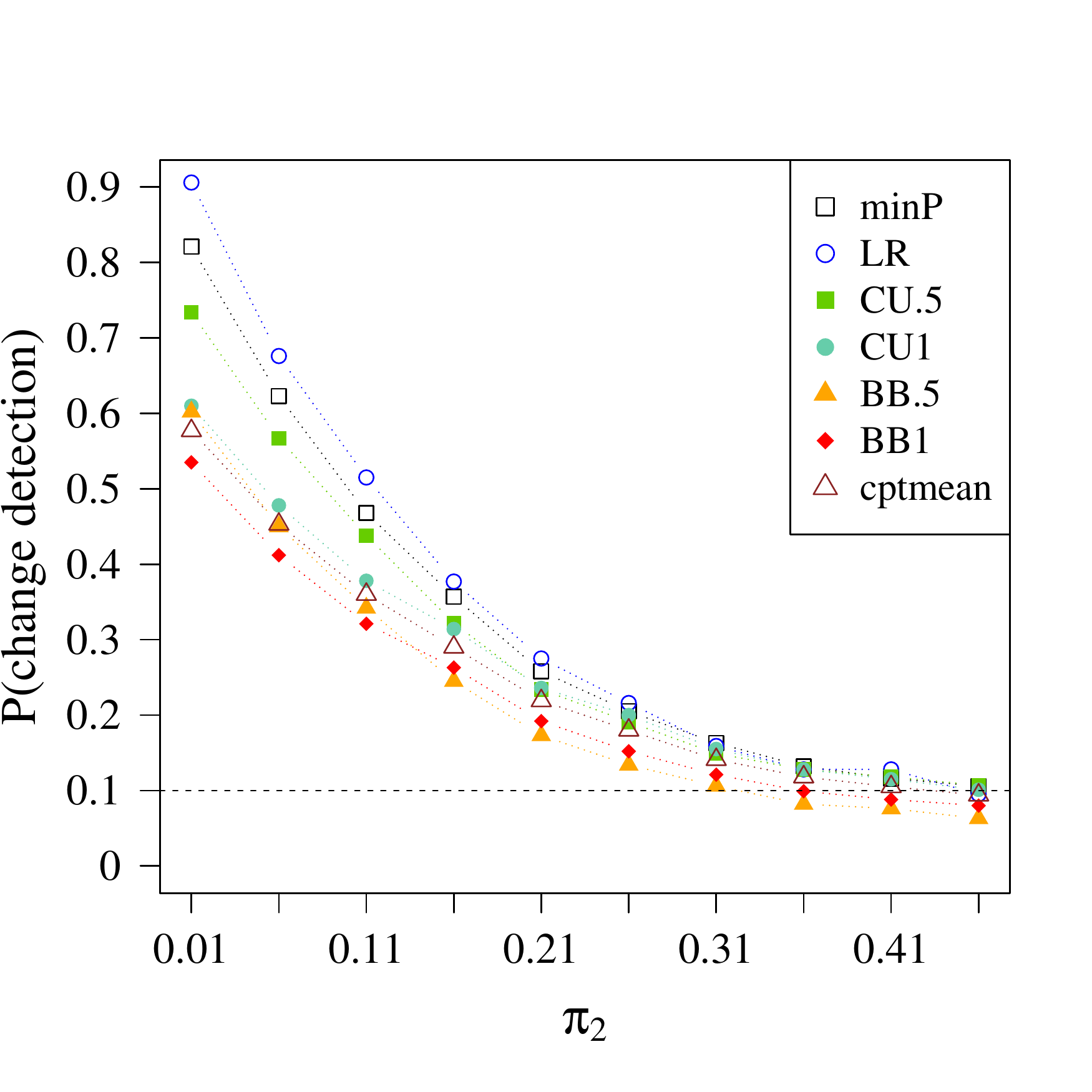} \\
        (a) $T=50, \tau=25, \pi_1 =0.4$ & (b) $T=50, \tau=40, \pi_1 =0.46$ \\
        \includegraphics[scale = 0.45]{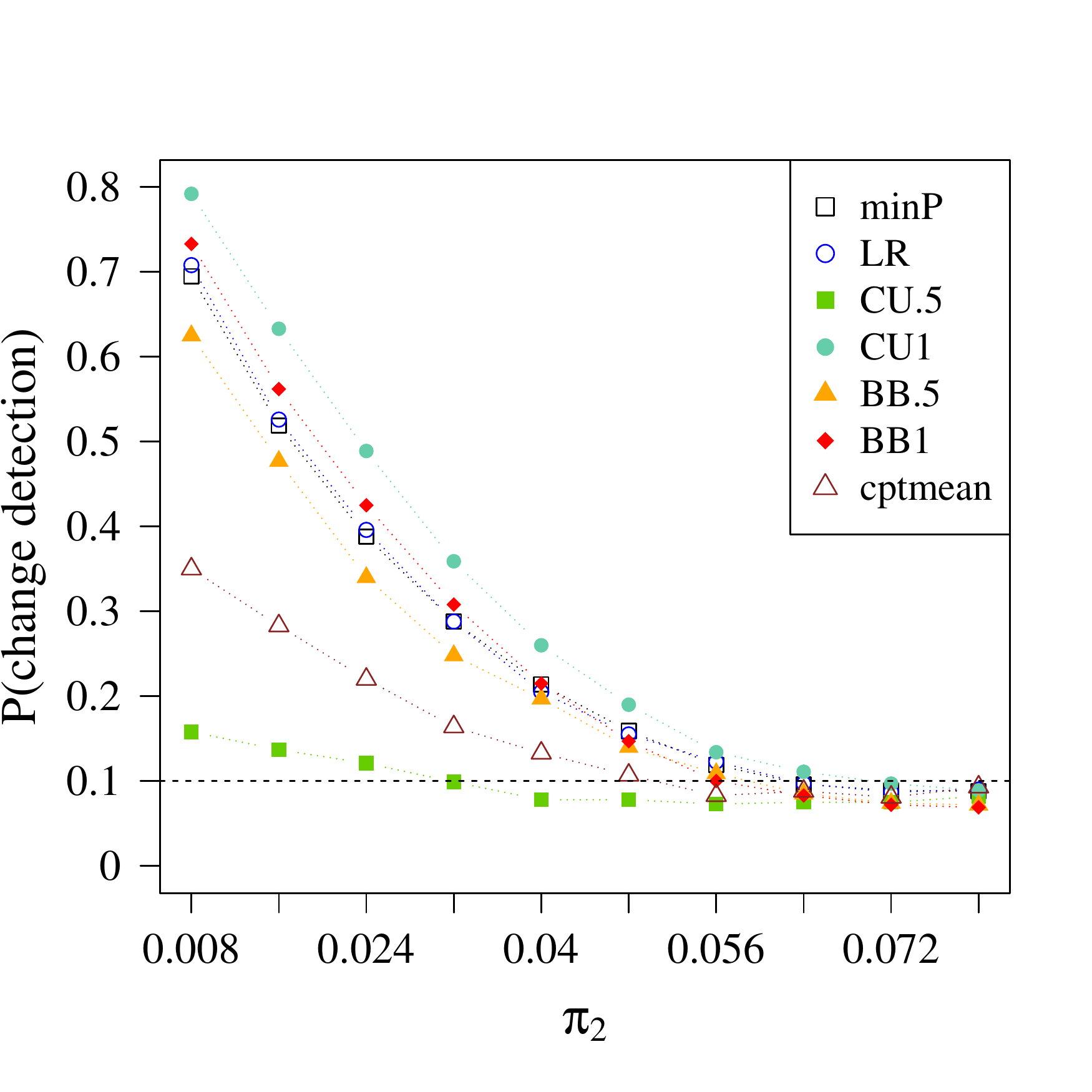} &
        \includegraphics[scale = 0.45]{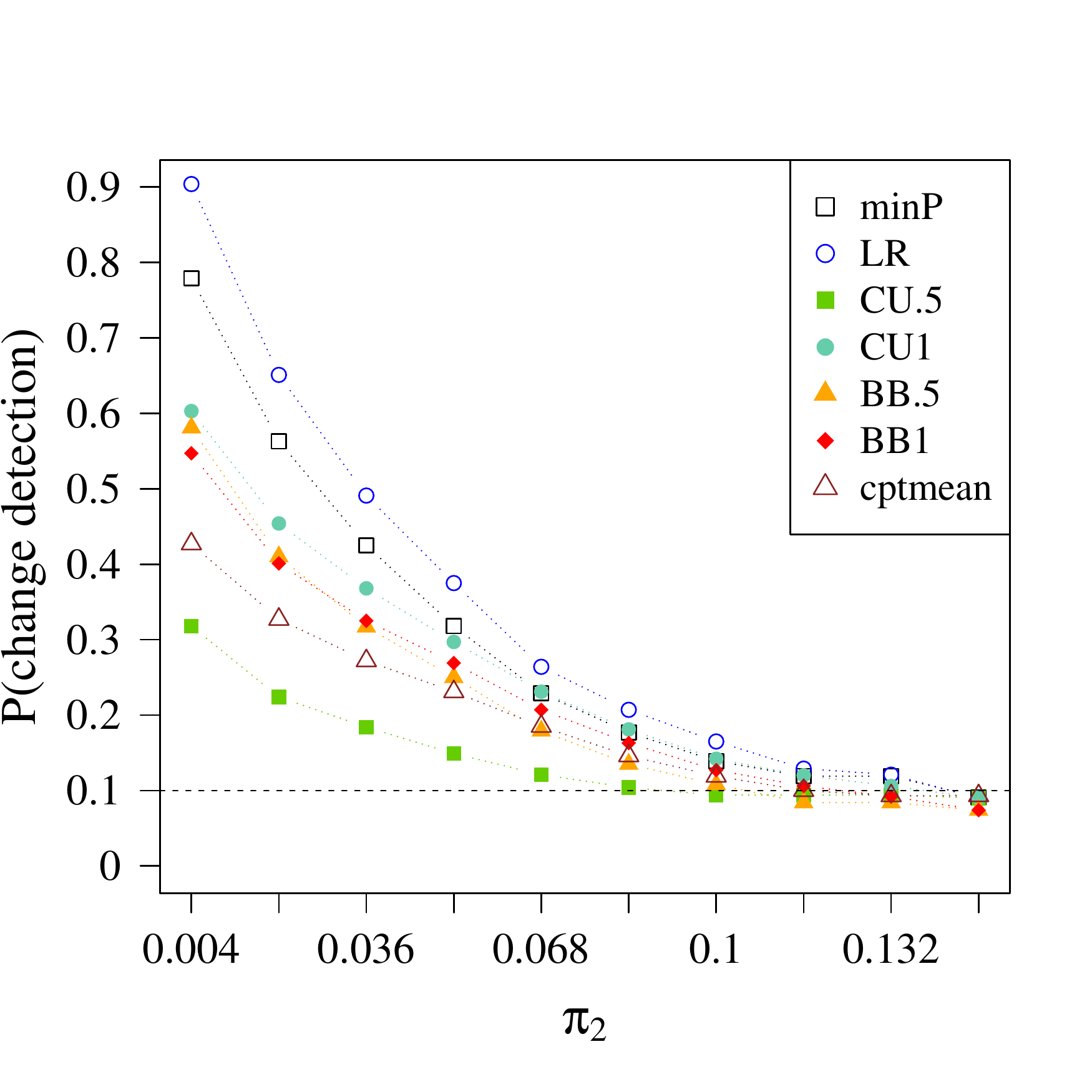} \\
        (c) $T=200, \tau=100, \pi_1 =0.08$ & (d) $T=200, \tau=160, \pi_1 =0.148$
    \end{tabular}
    \caption{Comparison of change detection probabilities of exact tests, asymptotic tests and the cptmean test with $\alpha=0.1$ in the time-series $X_1, \ldots, X_{\tau} \iid \bern(\pi_1)$, $X_{\tau+1}, \ldots, X_{T} \iid \bern(\pi_2)$.}
    \label{fig:bern.local.comp}
\end{figure}

\begin{figure}[!ht]
    \centering
    \begin{tabular}{cc}
        \includegraphics[scale = 0.45]{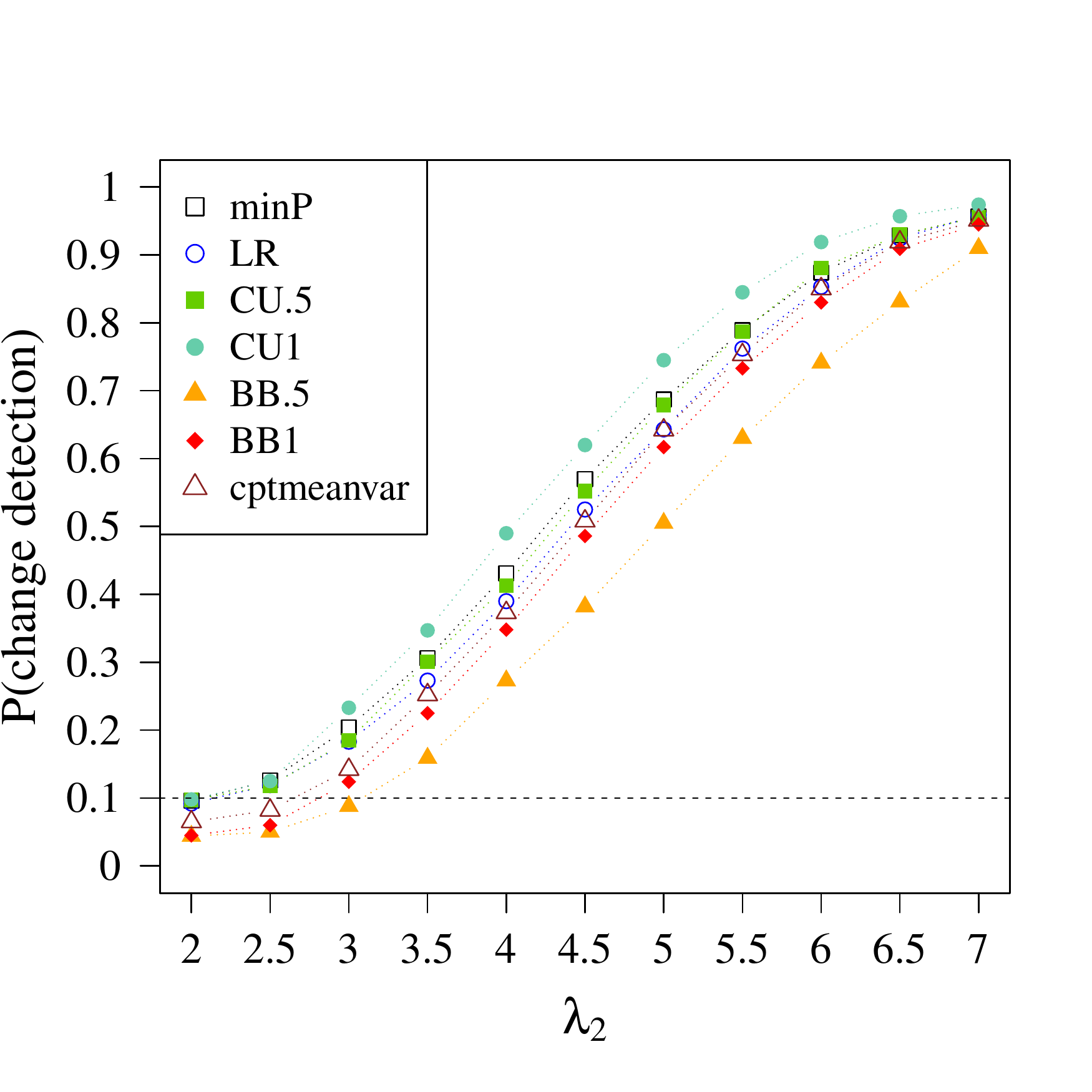} &
        \includegraphics[scale = 0.45]{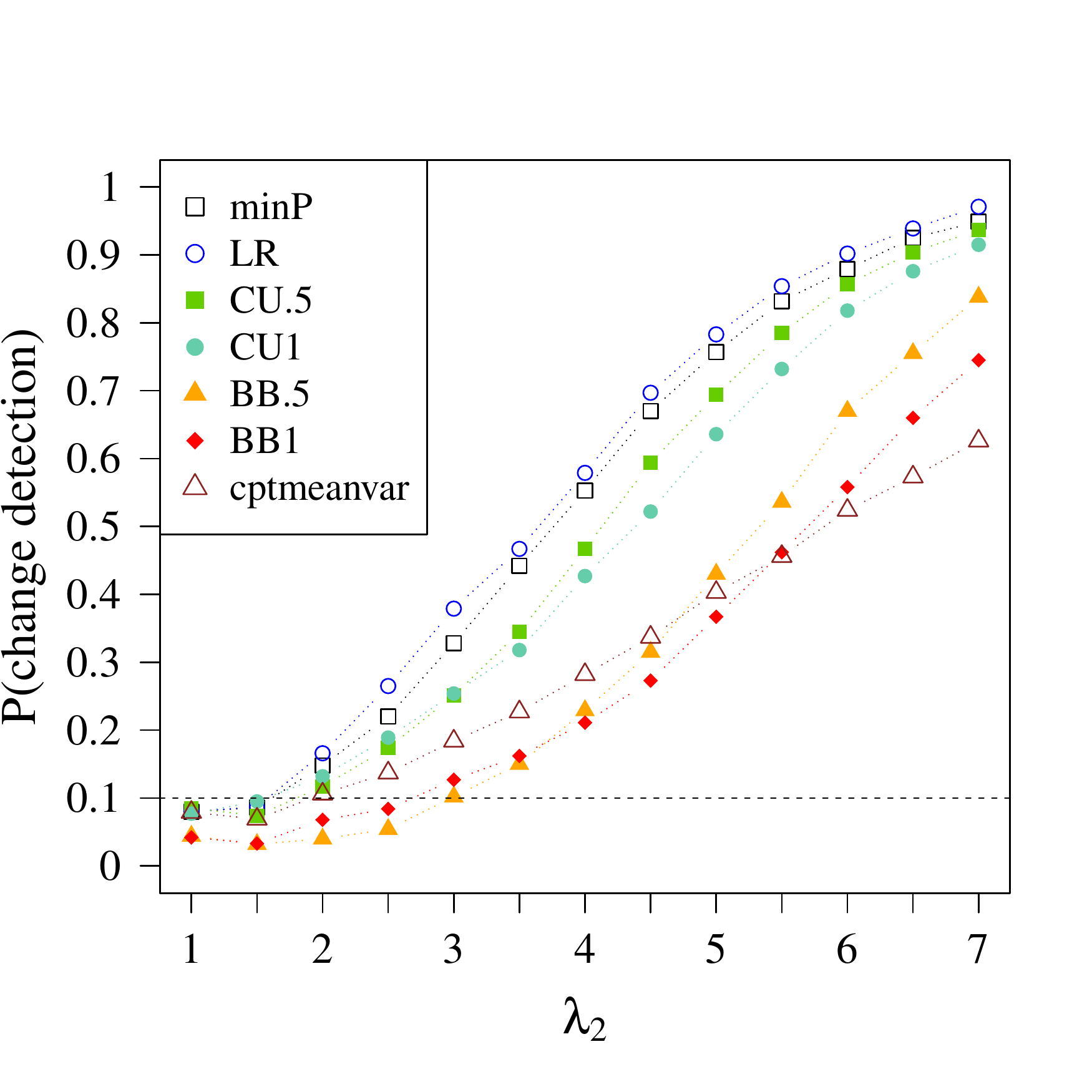} \\
        (a) $T=10, \tau=5, \lambda_1 =2$ & (b) $T=10, \tau=2, \lambda_1 =1$ \\
        \includegraphics[scale = 0.45]{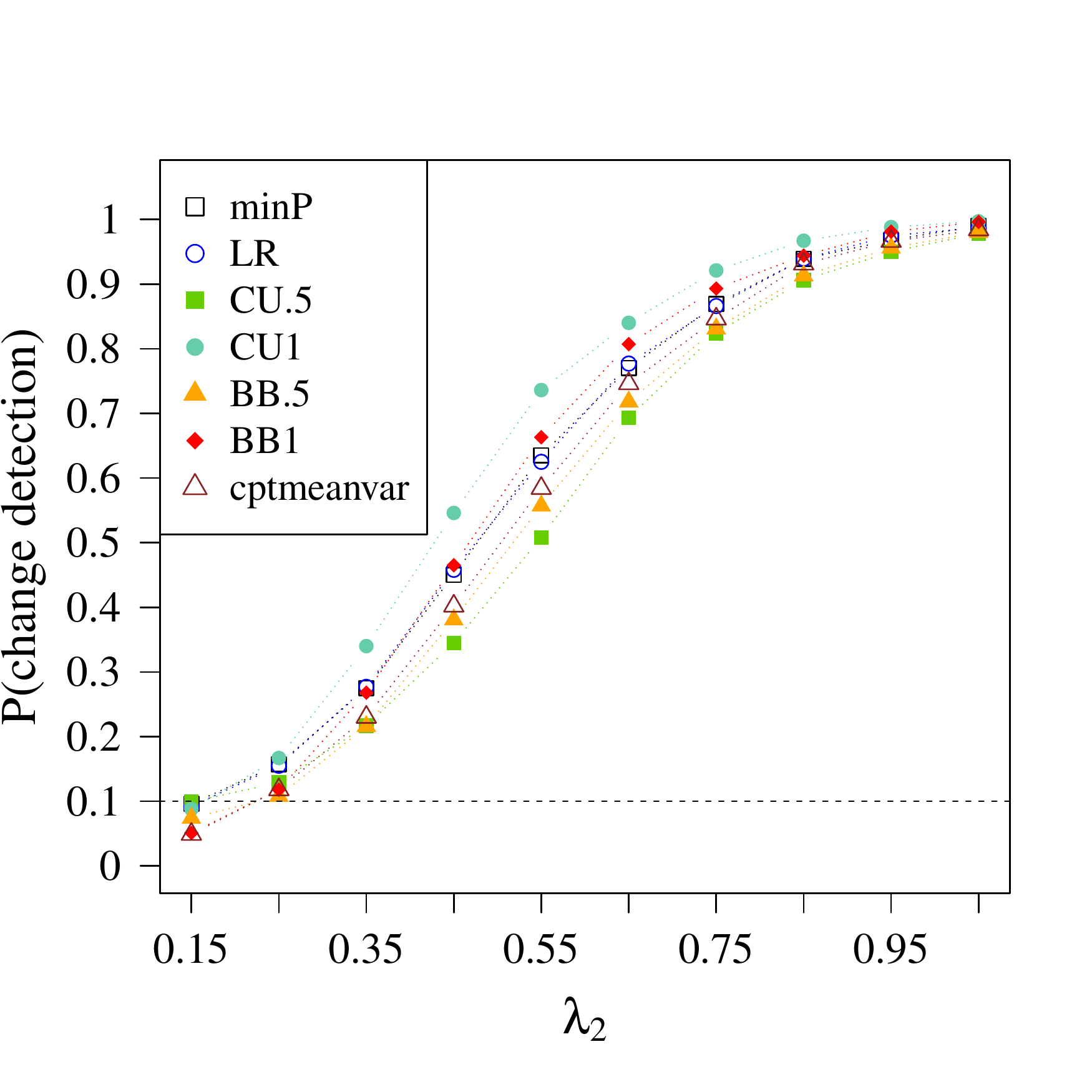} &
        \includegraphics[scale = 0.45]{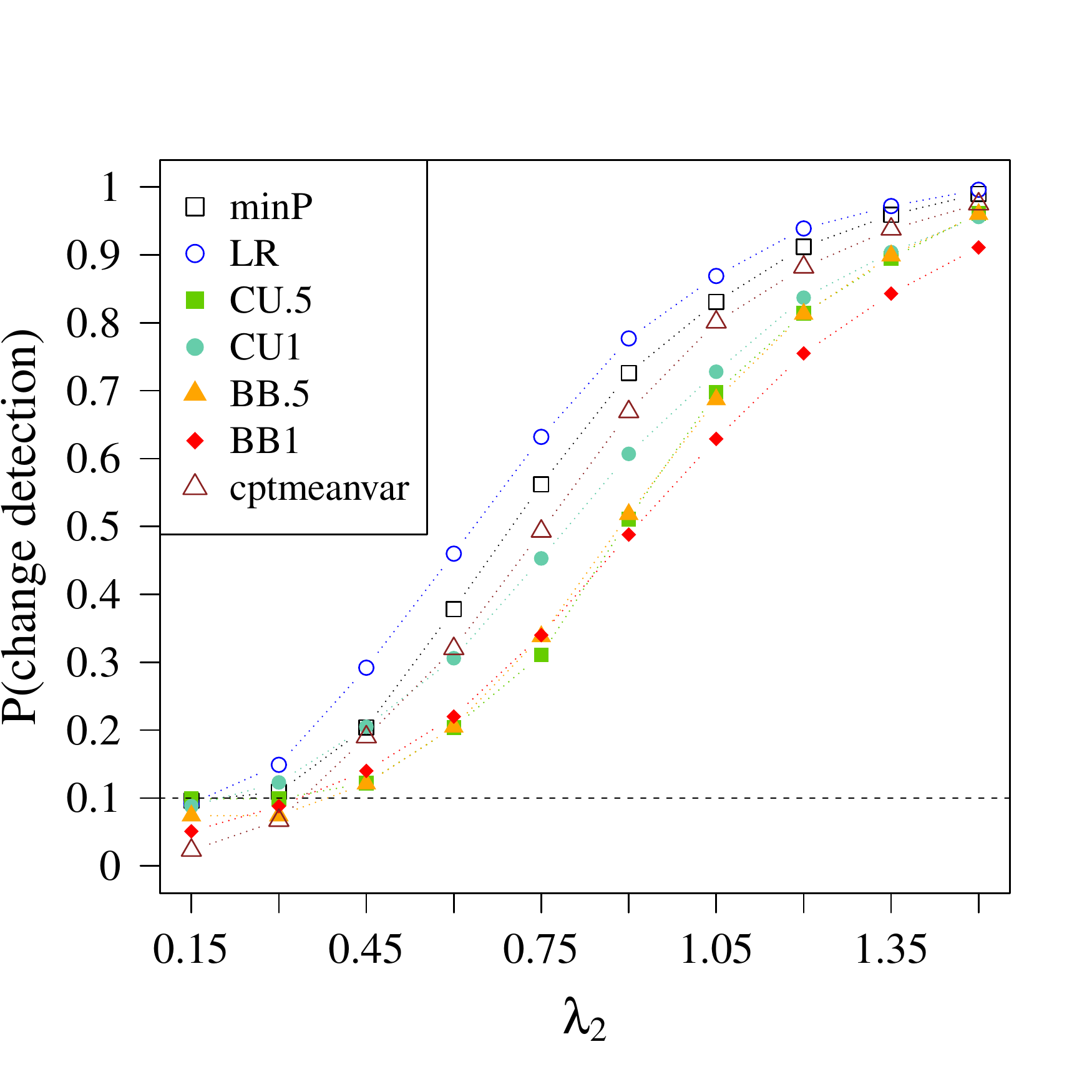} \\
        (c) $T=50, \tau=25, \lambda_1 =0.15$ & (d) $T=50, \tau=10, \lambda_1 =0.15$
    \end{tabular}
    \caption{Comparison of change detection probabilities of exact tests, asymptotic tests and the cptmeanvar test with $\alpha=0.1$ in the time-series: $X_1,\dots,X_{\tau} \iid \pois(\lambda_1)$, $X_{\tau+1},\dots,X_{T} \iid \pois(\lambda_2)$.}
    \label{fig:pois.local.comp}
\end{figure}

\subsection{Exact vs. asymptotic tests in a single channel}
\label{sec:exact-asymp}
We first compare the proposed exact level-$\alpha$ tests against the asymptotic level-$\alpha$ tests in a single channel. Exact conditional tests using the $\min_{1\le i \le T-1} p_i$ statistic, as discussed in Section \ref{sec:conditionaltests}, are referred to as ``minP'' tests, conditional tests based on $\T_{\lr}^{(b)}$ and $\T_{\lr}^{(c)}$ are referred to as ``LR'' tests, and  conditional tests based on the CUSUM statistics $\T_{\cusum}^{(0.5)}$ and $\T_{\cusum}^{(1)}$ are referred to as the ``CU.5'' test and the ``CU1'' test respectively. For these exact tests, $\alpha$-th quantiles of the respective test statistics under null are estimated from 50,000 Monte Carlo samples. 

Asymptotic tests based on Brownian bridge approximations (see Corollary~\ref{cor:BBapprox}) are considered for $\delta = 0.5$ and $\delta = 1$, and are referred to as the ``BB.5'' test and the ``BB1'' test respectively. 

Additionally, we consider two tests based on the functions \texttt{cpt.mean} and \texttt{cpt.meanvar} in the \textbf{R} package \texttt{changepoint}, which estimate the number of changepoints in univariate time-series. These are referred to as the ``cptmean'' test and the ''cptmeanvar'' test respectively. These tests\footnote{The ``cptmean'' test in Figure~\ref{fig:bern.local.comp} appiles the \texttt{cpt.mean} function with the ``BinSeg'' method and the CUSUM statistic. The ``cptmeanvar'' test in Figure~\ref{fig:pois.local.comp} applies the \texttt{cpt.meanvar} function with the ``BinSeg'' method and the ``Poisson'' statistic.} detect a change if the number of estimated changepoints is at least one.
 
Figure~\ref{fig:bern.local.comp} considers the Bernoulli case. We see that the exact conditional tests perform well in both sparse and dense situations and always outperform the cptmean test. The asymptotic tests (especially BB1) also provide reasonable power if the sample size $T$ is large and the changepoint $\tau$ is near the middle (Figures~\ref{fig:bern.local.comp}(a) and \ref{fig:bern.local.comp}(c)). However, if the changepoint is closer to the boundary (Figures~\ref{fig:bern.local.comp}(b) and \ref{fig:bern.local.comp}(d)), then the exact conditional tests minP and LR perform significantly better than the asymptotic tests.

Figure~\ref{fig:pois.local.comp} considers the Poisson case. The proposed exact tests perform well even when the sample size is as small as $T=10$, and, in this case, they uniformly outperform the asymptotic tests and the cptmeanvar test. If the changepoint is close to the boundary, then the exact tests (especially minP and LR) yield much higher power than their competitors (see Figures~\ref{fig:pois.local.comp}(b) and \ref{fig:pois.local.comp}(d)). For large sample sizes (e.g., $T=50$), when the Brownian bridge approximations kick in, asymptotic tests become comparable to the exact tests in terms of performance.

\subsection{Global vs. local testing in multiple channels}
\label{sec:glob-loc}
Global testing of $H_0$ is done by permutation tests using $C^{(\delta)}$ with Euclidean norm as discussed in Section~\ref{sec:mult}. The $m$-variate time-series $\boldsymbol{X}_1,\dots, \boldsymbol{X}_T$ is permuted $B=1000$ times to obtain a randomized size-$\alpha$ test. For power comparisons, two tests ``gCU.5'' and ``gCU1'' are considered that are obtained using the global CUSUM statistics $C^{(0.5)}$ and $C^{(1)}$ respectively.

To test each channel for possible changepoints, we consider three exact conditional tests, namely minP, LR and CU1. After computing $p$-values from these tests, we employ the BH procedure to obtain $\mathcal{R} = \{1\le j \le m: H_{0, j} \text{ is rejected} \}$. Henceforth, we refer to these local tests as minP-BH, LR-BH and CU1-BH respectively. Figures~\ref{fig:berpowglobloc} and \ref{fig:poispowglobloc} compare probabilities of global change detection (gCD), i.e. probabilities of rejecting $H_0$ (this is $P(\mathcal{R} \neq \emptyset)$ for local tests) for global and local tests. Figure~\ref{fig:berpowglobloc} considers Bernoulli channels while Figure~\ref{fig:poispowglobloc} deals with Poisson channels. We find that the local tests are significantly more powerful than the global tests in the rare signal regime where $n_{\text{cp}}$ is small or moderate. Also, the power advantage is more and continues over a longer range of $n_{\mathrm{cp}}$ when the changepoint is near the boundary. The local and global tests have comparable power for large $n_{\text{cp}}$, as expected. 

\begin{figure}[!ht]
    \centering
    \begin{tabular}{cc}
        \includegraphics[scale = 0.45]{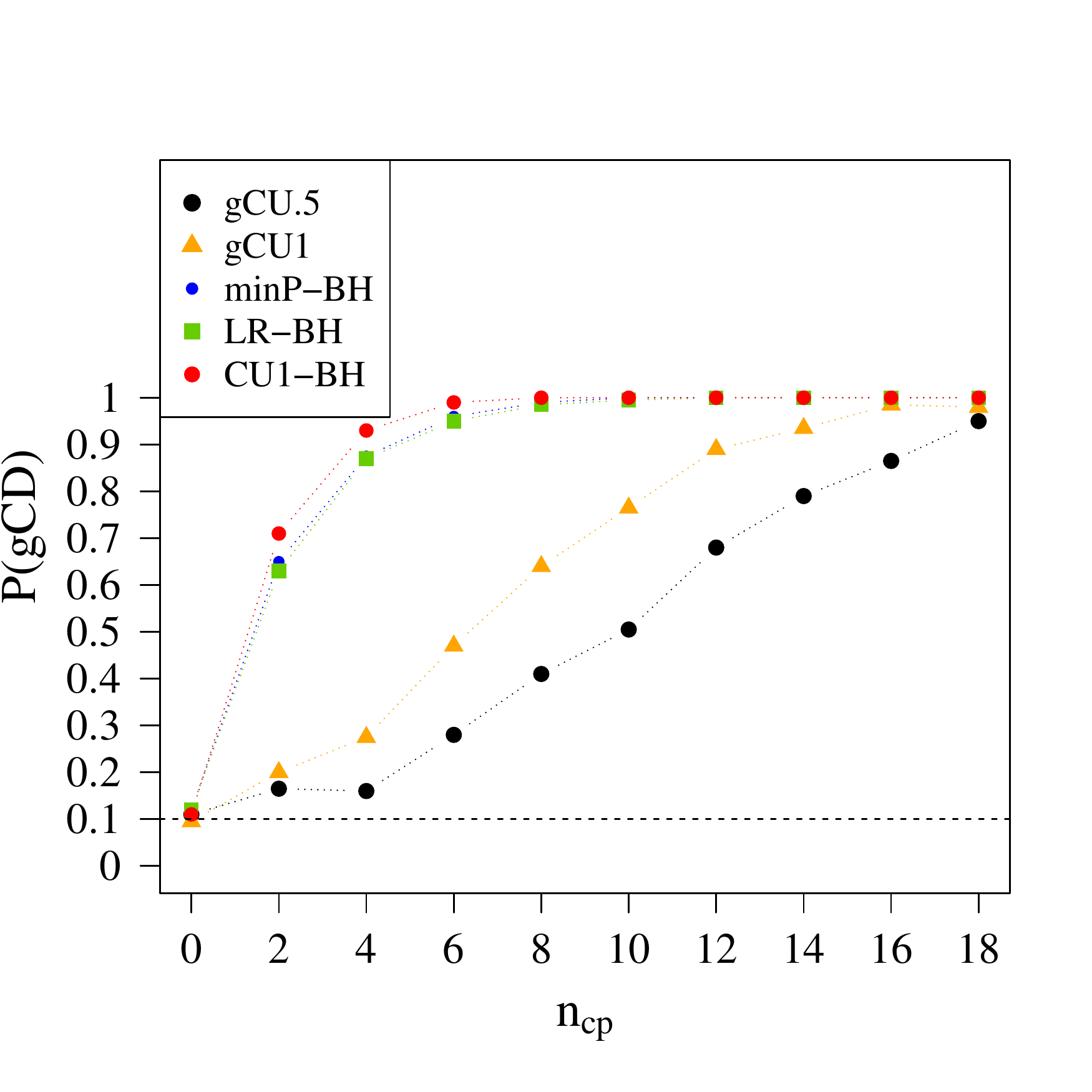} &
        \includegraphics[scale = 0.45]{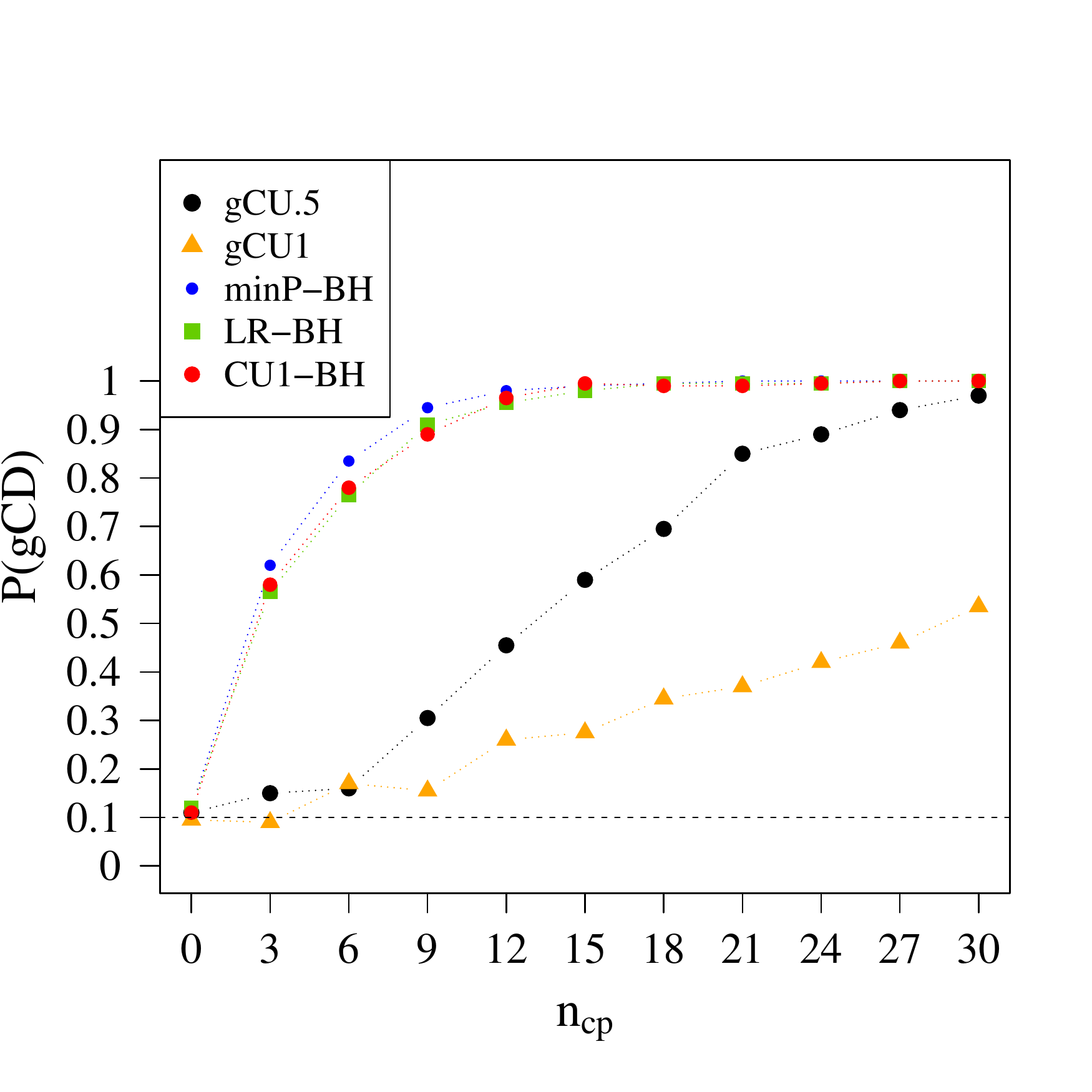} \\
        (a) $\tau=100$ & (b) $\tau=160$
    \end{tabular}
    \caption{Comparison of P(gCD) of global and local tests in $m=1000$ independent Bernoulli series: $X_{j, 1}, \ldots, X_{j, \tau} \iid \bern(\pi_1 = 0.05)$,\, $X_{j, \tau+1}, \ldots, X_{j, T} \iid \bern(\pi_2 = 0.25)$ where $T=200$. Tests are conducted at level $\alpha=0.1$ and $n_{\text{cp}}$ channels undergo change at time-point $\tau$.}
    \label{fig:berpowglobloc}
\end{figure}

\begin{figure}[!ht]
    \centering
    \begin{tabular}{cc}
        \includegraphics[scale = 0.45]{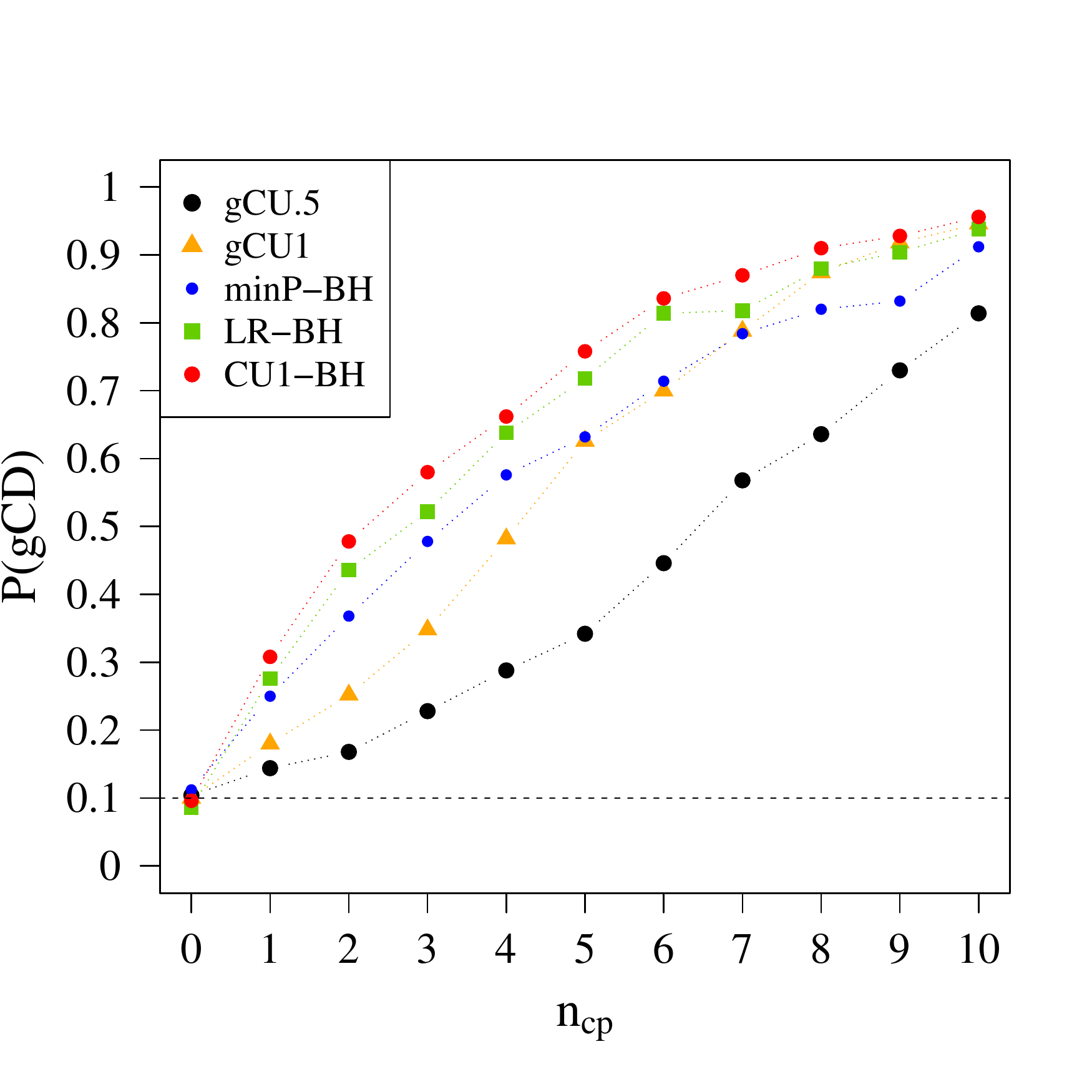} &
        \includegraphics[scale = 0.45]{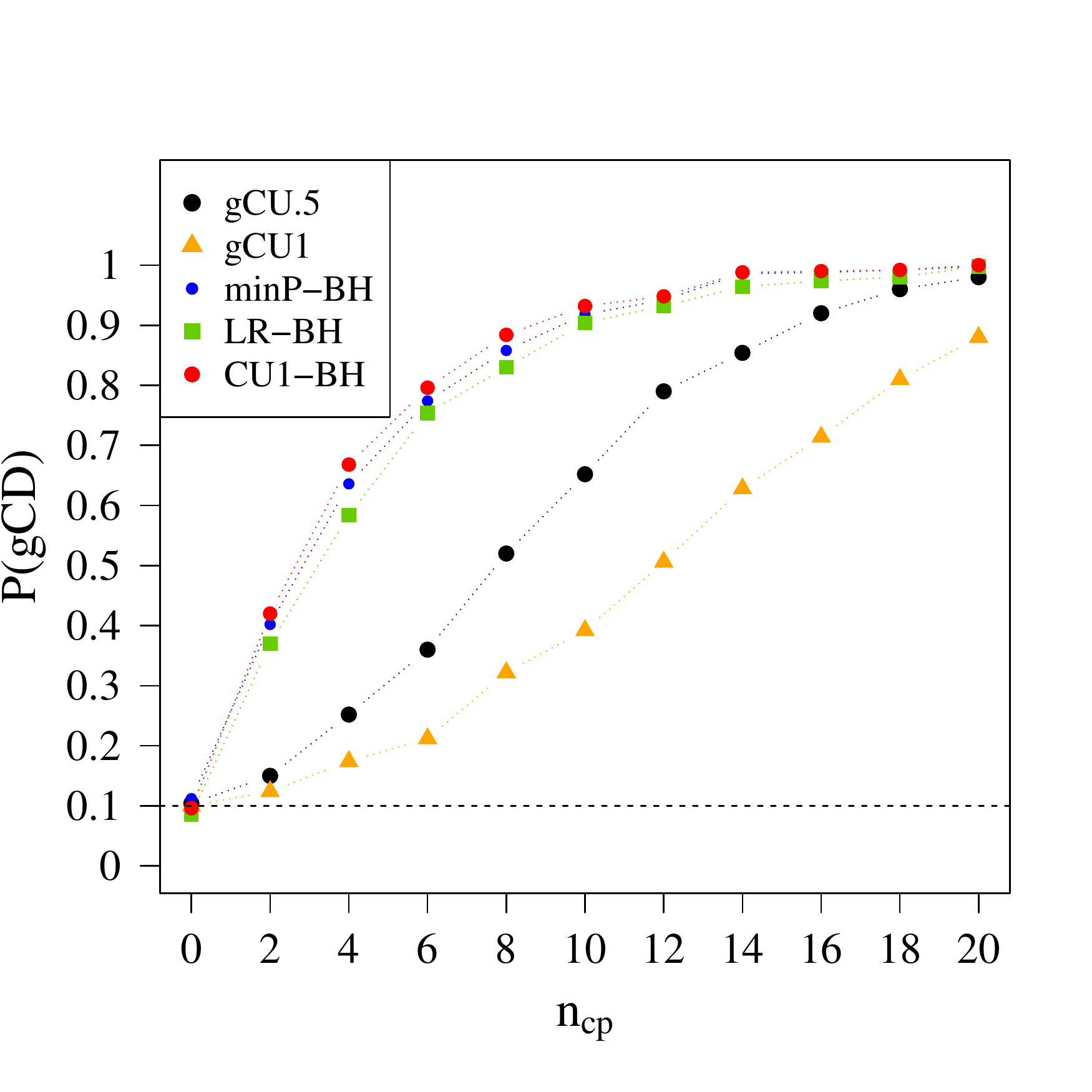} \\
       (a) $\tau=10$ & (b) $\tau=15$
    \end{tabular}
    \caption{Comparison of P(gCD) of global and local tests in $m=200$ independent Poisson series: $X_{j, 1}, \ldots, X_{j, \tau} \iid \pois(\lambda_1=0.25)$,\, $X_{j, \tau+1}, \ldots, X_{j, T} \iid \pois(\lambda_2=1.5)$ where $T=20$. Tests are conducted at level $\alpha=0.1$ and $n_{\text{cp}}$ channels undergo change at time-point $\tau$.}
    \label{fig:poispowglobloc}
\end{figure}

In a simulation study presented in the appendix, we consider two additional FDR controlling methods, namely the adaptive Benjamini-Hocheberg (ABH) and the adaptive Storey-Taylor-Siegmund (STS) methods. Performances of these methods are comparable to that of the vanilla BH method as the simulation study is done under the rare signal regime. Both the ABH and the STS methods are implemented using the \textbf{R} package \texttt{mutoss} \citep{mutoss17}.

\section{Real data}
\label{sec:realdata}
Now we analyse two datasets that can be naturally summarised by networks. These give real examples of multichannel binary and count data with potential changepoints.

\subsection{US senate rollcall data}
\label{sec:USvote}
From the US senate rollcall dataset \citep{voteview2020}, we construct networks where nodes represent US senate seats. Each epoch represents a proposed bill on which votes were taken. An edge between two seats is formed if they voted similarly on that bill. We have $n=100$ nodes. We consider $T = 50$ time-points between August 10, 1994 and January 24, 1995.

There are $m = \binom{n}{2} = 4950$ channels (i.e. edges). Of these, $622$ channels are ignored while analyzing this data since those channels contain too many zeros or ones (more than 45). We applied the BH procedure to simultaneously test the remaining $4328$ channels controlling FDR at level $\alpha = 0.05$. For each significant channel, the corresponding changepoint location is also reported (see the discussion in Section~\ref{sec:estimation}).

In Figure~\ref{fig:uss_hist}, we plot the histograms of the changepoint locations of the significant channels. Note particularly the peak near time-point $24$ (which corresponds to December 1, 1994). There is a historically well-documented change near December 1994, which saw the end of the conservative coalition (see, e.g., \cite{moody2013portrait}). Interestingly, the global method also detected a changepoint at $t = 24$. Changepoints at nearly the same location were found earlier in \cite{roy2017change} and \cite{mukherjee2018thesis}. However, the local methods have the advantage of identifying the channels that underwent a change. The number of significant channels, $n_{s}$, is reported below each histogram. A number of channels had extremely small $p$-values. For example, Figure~\ref{fig:uss_twocha}(a) depicts the time-series of edges $(4, 5)$ and $(4, 6)$. Changes are visible to the naked eye.  Seats $5$ and $6$ are in Arizona, while seat 4 is in Arkansas. Clearly, seat 4 went from agreeing with seats 5 and 6 to disagreeing. On the other hand, seat 3 is also from Arkansas, and no changepoints were found in the channels $(3, 5)$ and $(3, 6)$ (see Figure~\ref{fig:uss_twocha}(b)).

\begin{figure}[!ht]
    \centering
    \begin{tabular}{ccc}
        \includegraphics[scale = 0.28]{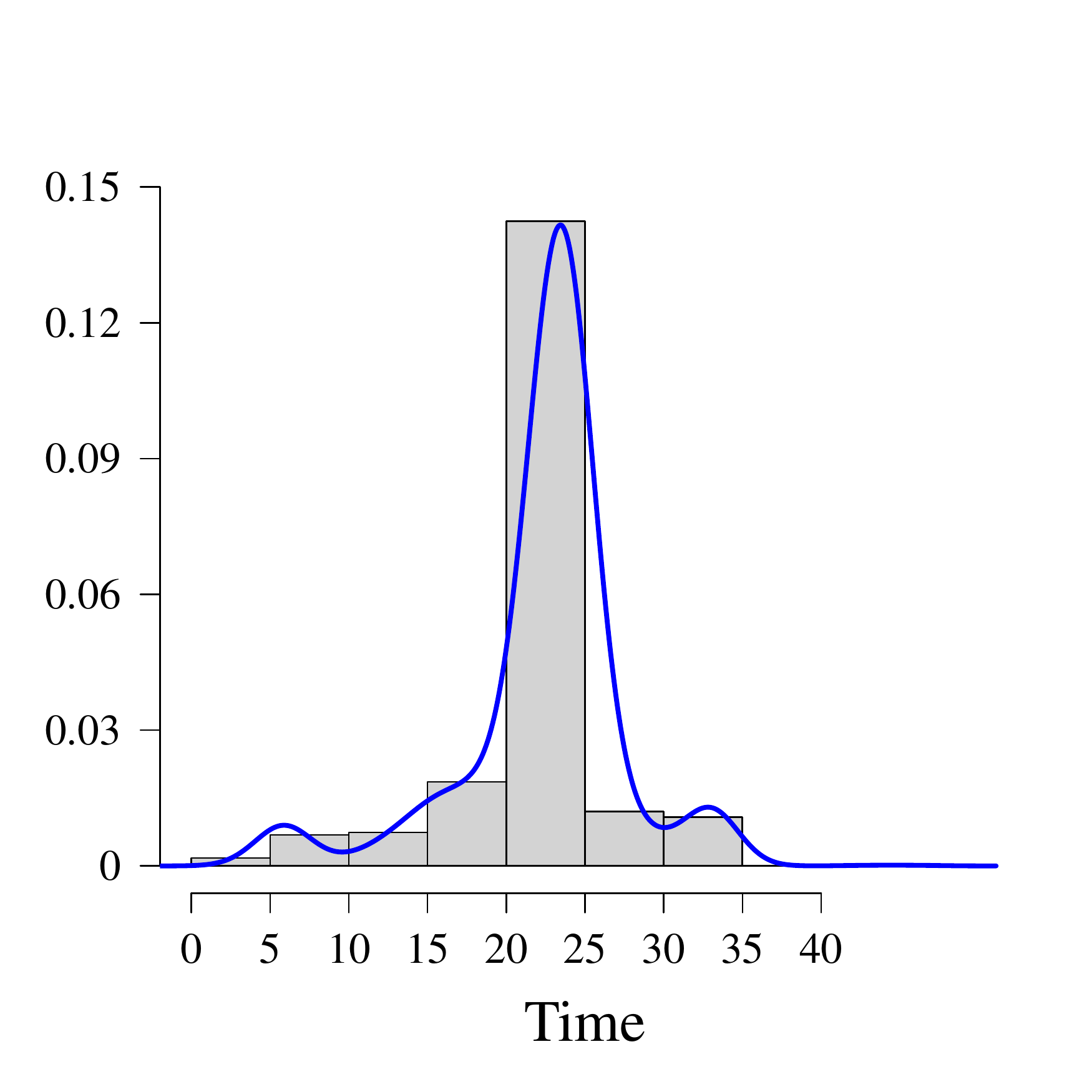} &
        \includegraphics[scale = 0.28]{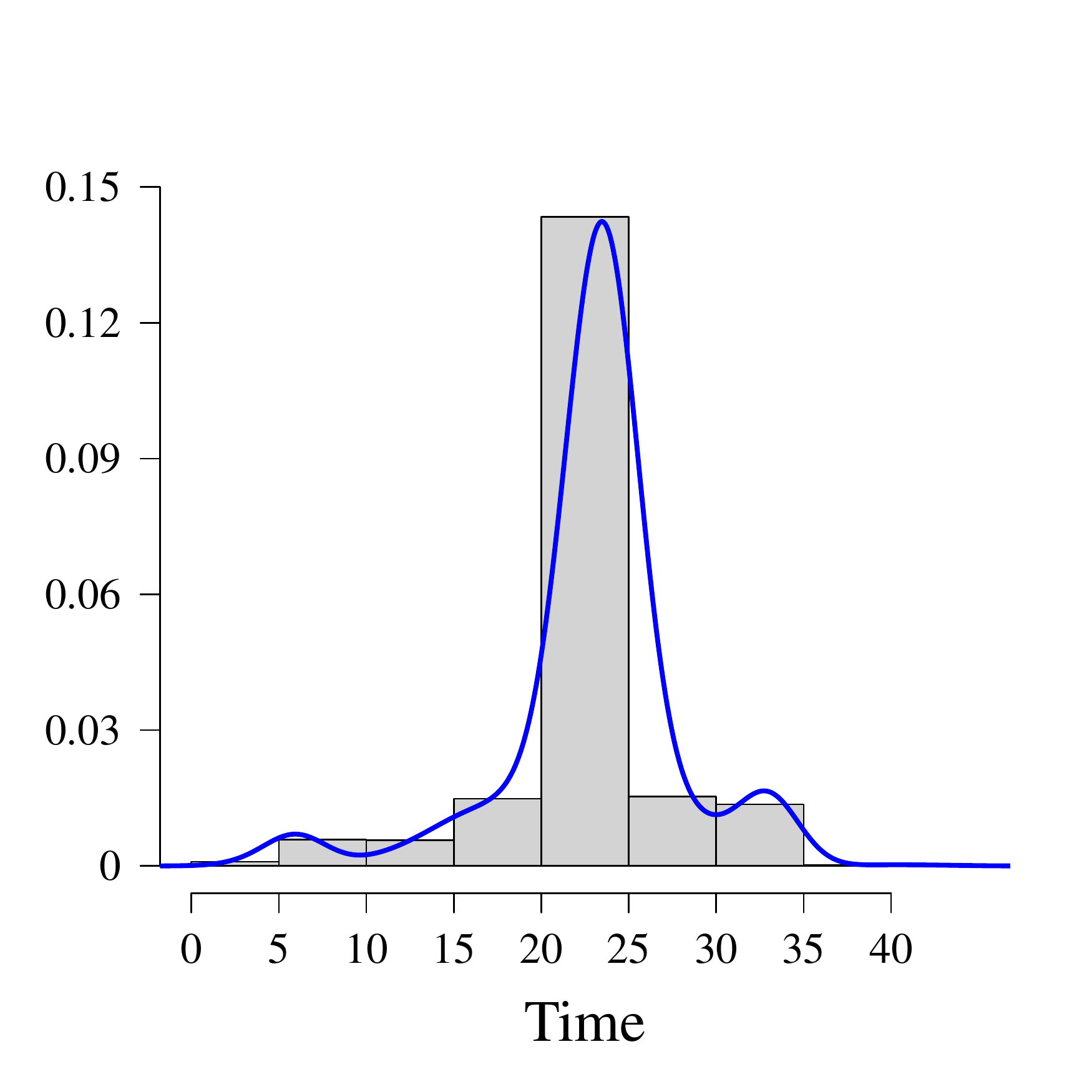} &
        \includegraphics[scale = 0.28]{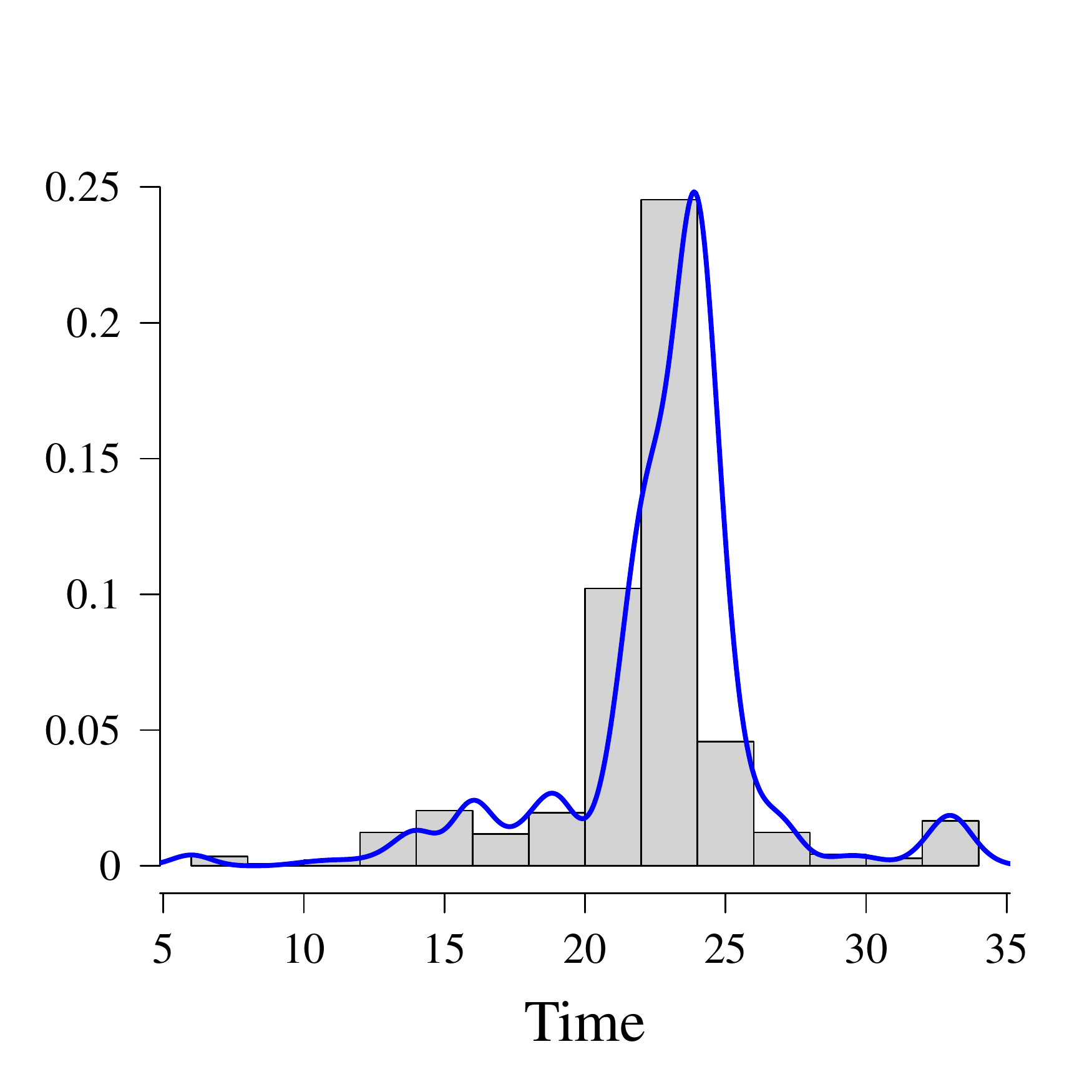} \\
        (a) minP-BH, $n_{s}$ = 1950 & (b) LR-BH, $n_{s}$ = 1980 & (c) CU1-BH, $n_{s}$ = 1987
    \end{tabular}
    \caption{Distribution of detected changepoint locations in the US senate rollcall data. All three methods report a mode at $t = 24$.}
    \label{fig:uss_hist}
\end{figure}
\begin{figure}[!ht]
    \centering
    \begin{tabular}{cc}
        \includegraphics[scale = 0.45]{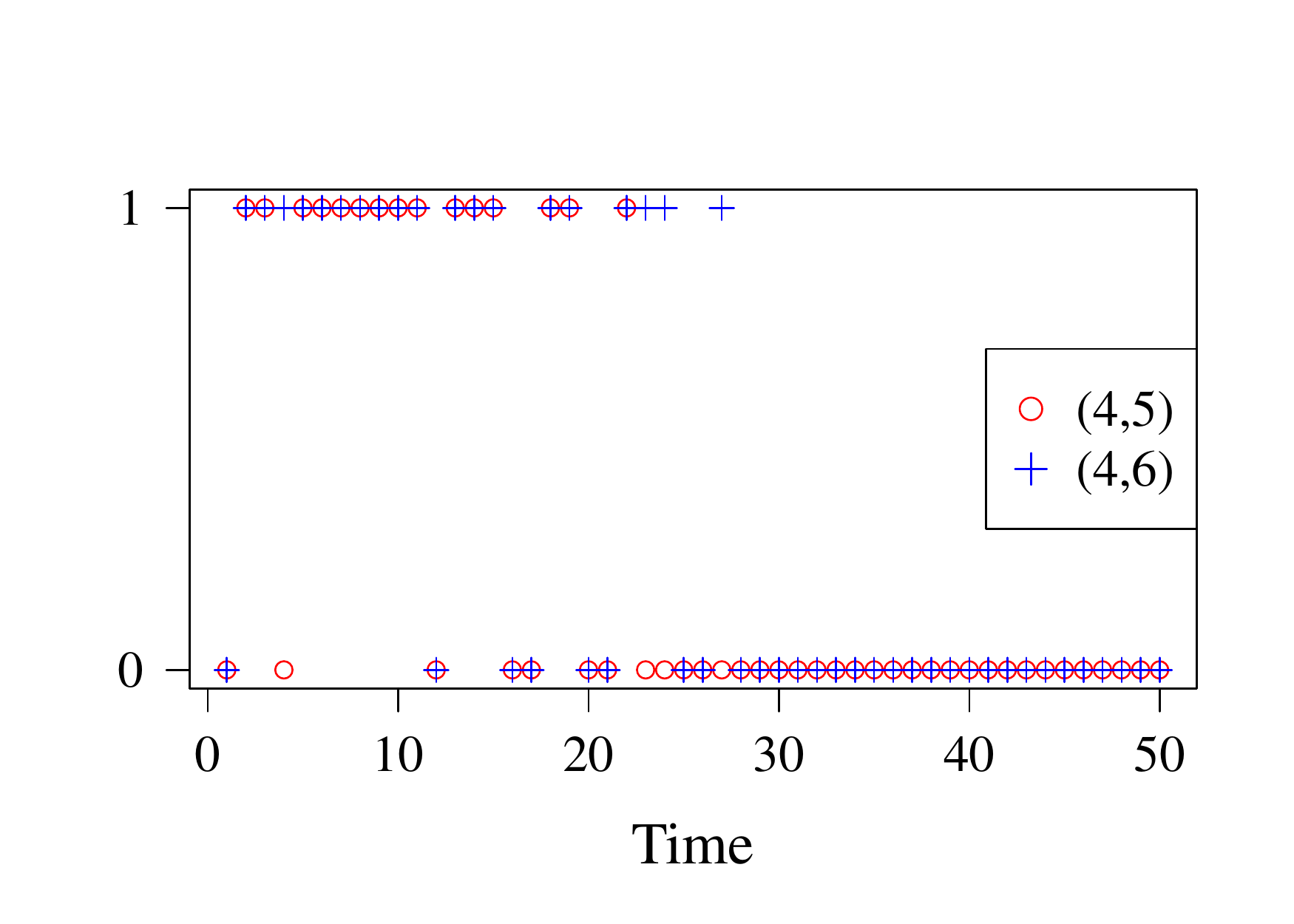} &
        \includegraphics[scale = 0.45]{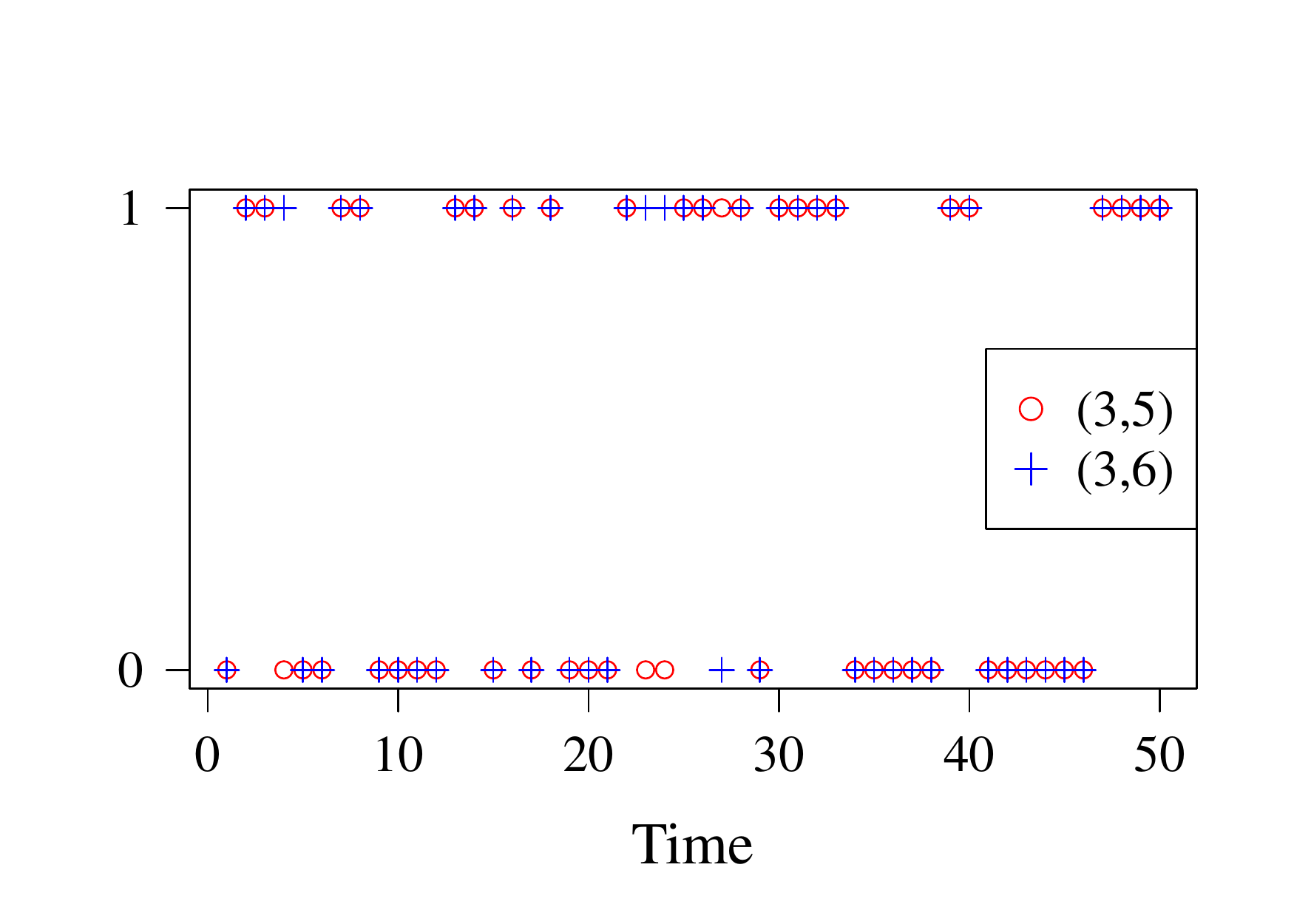} \\
    (a) & (b) 
    \end{tabular}
    \caption{(a) Edges $(4, 5)$ and $(4, 6)$. Seats $5$ and $6$ are in Arizona, while seat 4 is in Arkansas. Clearly, seat 4 went from agreeing with seats 5 and 6 to disagreeing. (b) Edges $(3, 5)$ ans $(3, 6)$. Seat $3$ is also in Arkansas. No changepoints are detected in these channels.}%
    \label{fig:uss_twocha}
\end{figure}

\subsection{MIT reality mining data}
We use the MIT reality mining data \citep{eagle2006reality} to construct a series of networks involving $n = 90$ individuals (staff and students at the university). The data consists of call logs between these individuals from 20th July 2004 to 14th June 2005. We construct $T = 48$ weekly networks, where a weighted edge between nodes $u$ and $v$ reports the number of phone calls between them during the corresponding week. There are $m = \binom{n}{2} = 4005$ channels (i.e. edges). 3945 channels are ignored while analyzing this data since those channels contain too many zeros (more than 44). The remaining 60 channels are tested for possible changepoints.

We model the weighted edges as Poisson variables and apply the exact tests minP, LR or CU1 on each channel. Then we apply the BH method to simultaneously test the 60 channels controlling FDR at level $\alpha = 0.05$. Figure~\ref{fig:mit_hist} contains the histograms of the detected changepoint locations. Note particularly the peaks near $t = 20$ and $t = 33$. For comparison, a changepoint at $t = 24$ was found by a global algorithm in \cite{mukherjee2018thesis}. The graph-based multivariate (global) change detection methods of \cite{chen2015} found changepoints at approximately $t = 22$ and $t = 25$ (their analyses were on daily networks). The global algorithms consider global characteristics, and thus it is not surprising that they find changepoints somewhat in the middle of the predominant local changepoints near $t = 20$ and $t = 33$. It turns out that $t = 20$ is just before the start of the Winter break, and $t = 33$ is just before the start of the Spring break.

We also perform changepoint analysis with node-degrees as channels, modeled as a Poisson series. Figure~\ref{fig:mitdeg_hist} shows the histograms of the detected changepoint locations. Analyses of edge and degree time-series detect 44 and 46 common nodes (i.e. detected by all three tests: minP-BH, LR-BH and CU1-BH) respectively. Among these, $40$ nodes are declared significant by both analyses.

Finally, in Figure~\ref{fig:mit_graph}, we show the average networks before and after time-point 20. The 46 channels (i.e. nodes) declared to have changepoints by all three tests under the degree-based analysis are shown as green circles. A structural change is clearly visible.

\begin{figure}[!ht]
    \centering
    \begin{tabular}{ccc}
        \includegraphics[scale = 0.28]{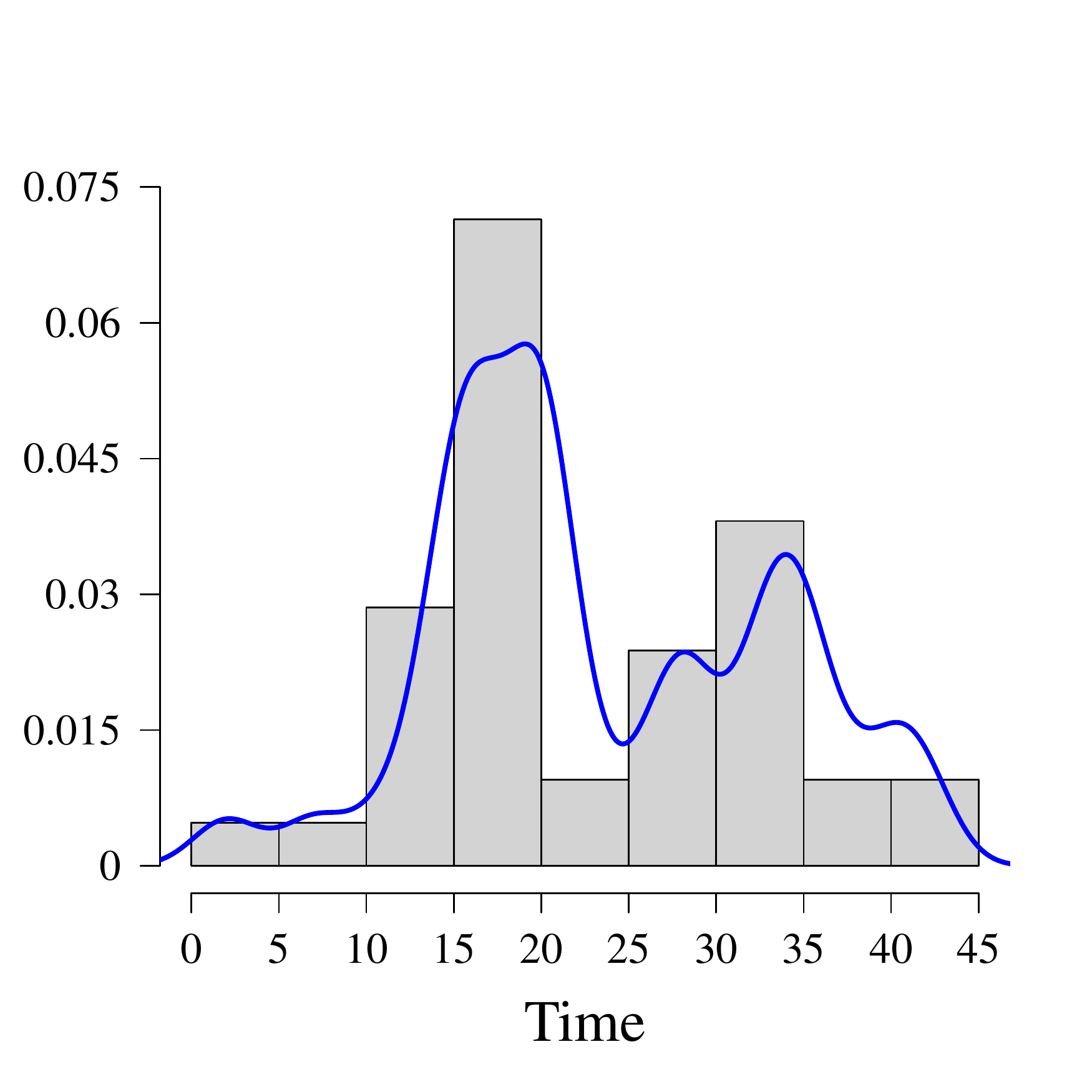} &
        \includegraphics[scale = 0.28]{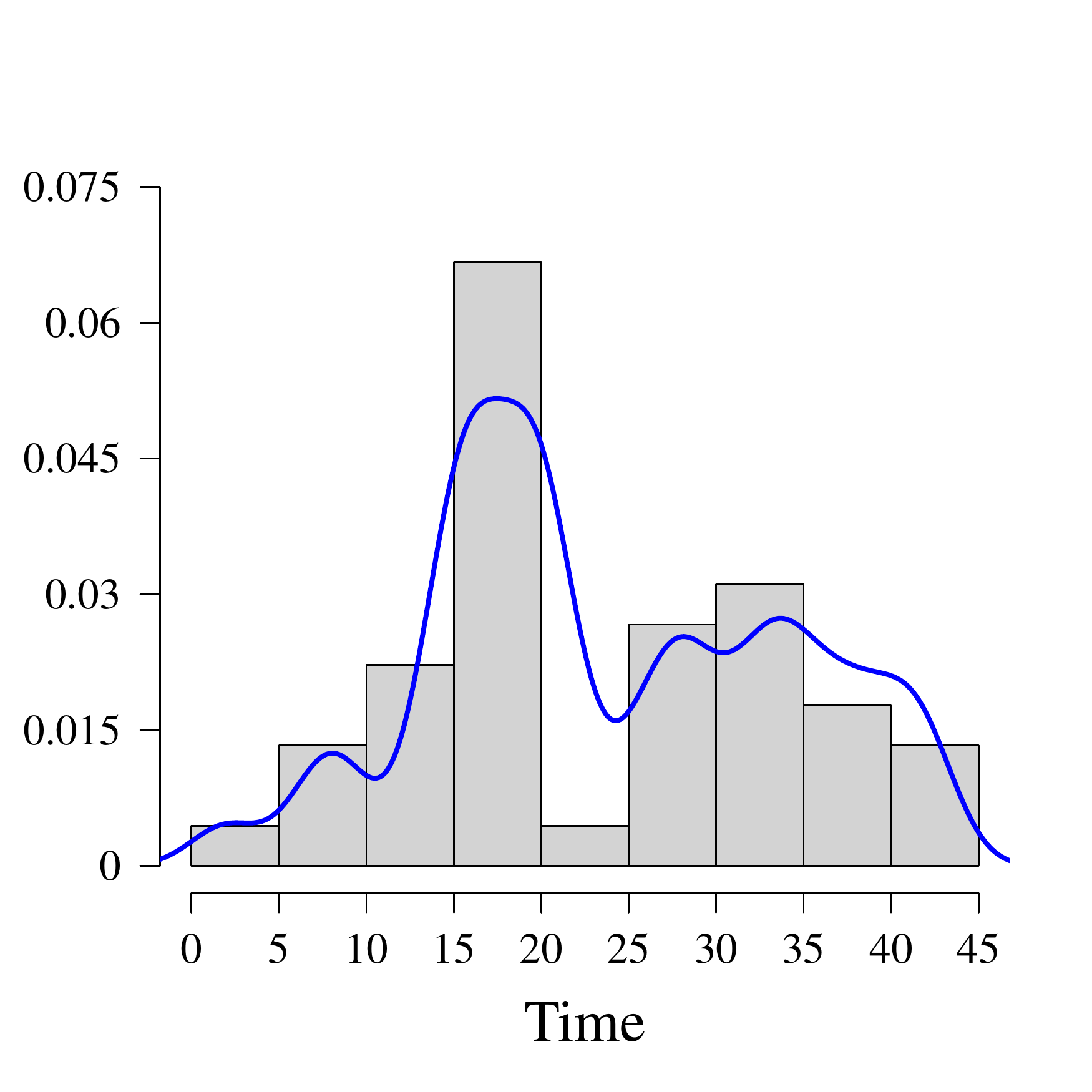} &
        \includegraphics[scale = 0.28]{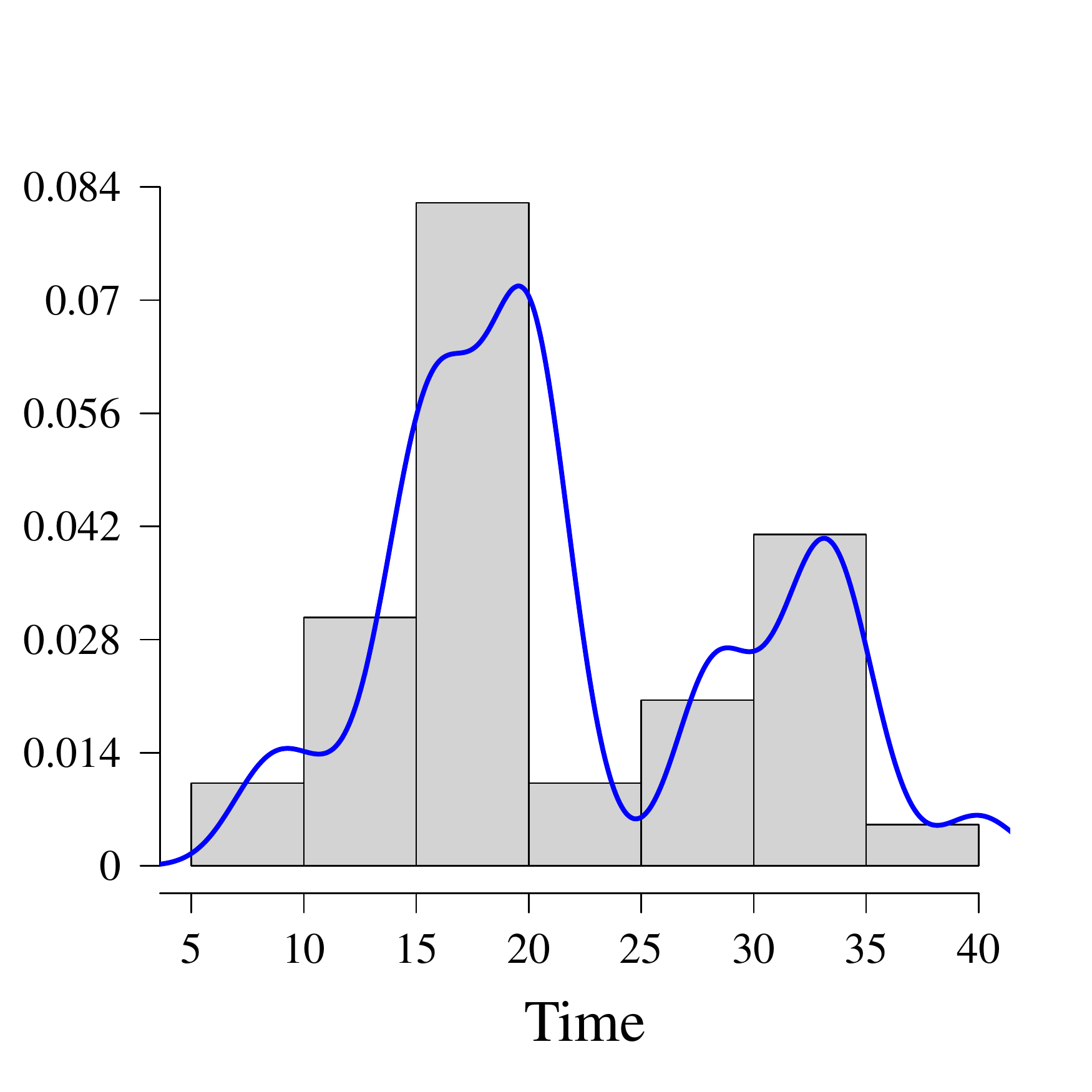} \\
        (a) minP-BH, $n_{s}=42$ & (b) LR-BH, $n_{s}=45$ & (c) CU1-BH, $n_{s}=39$
    \end{tabular}
    \caption{Changepoint locations from edge-based analysis of the MIT reality mining data.}
    \label{fig:mit_hist}
\end{figure}
\begin{figure}[!ht]
    \centering
    \begin{tabular}{ccc}
        \includegraphics[scale = 0.28]{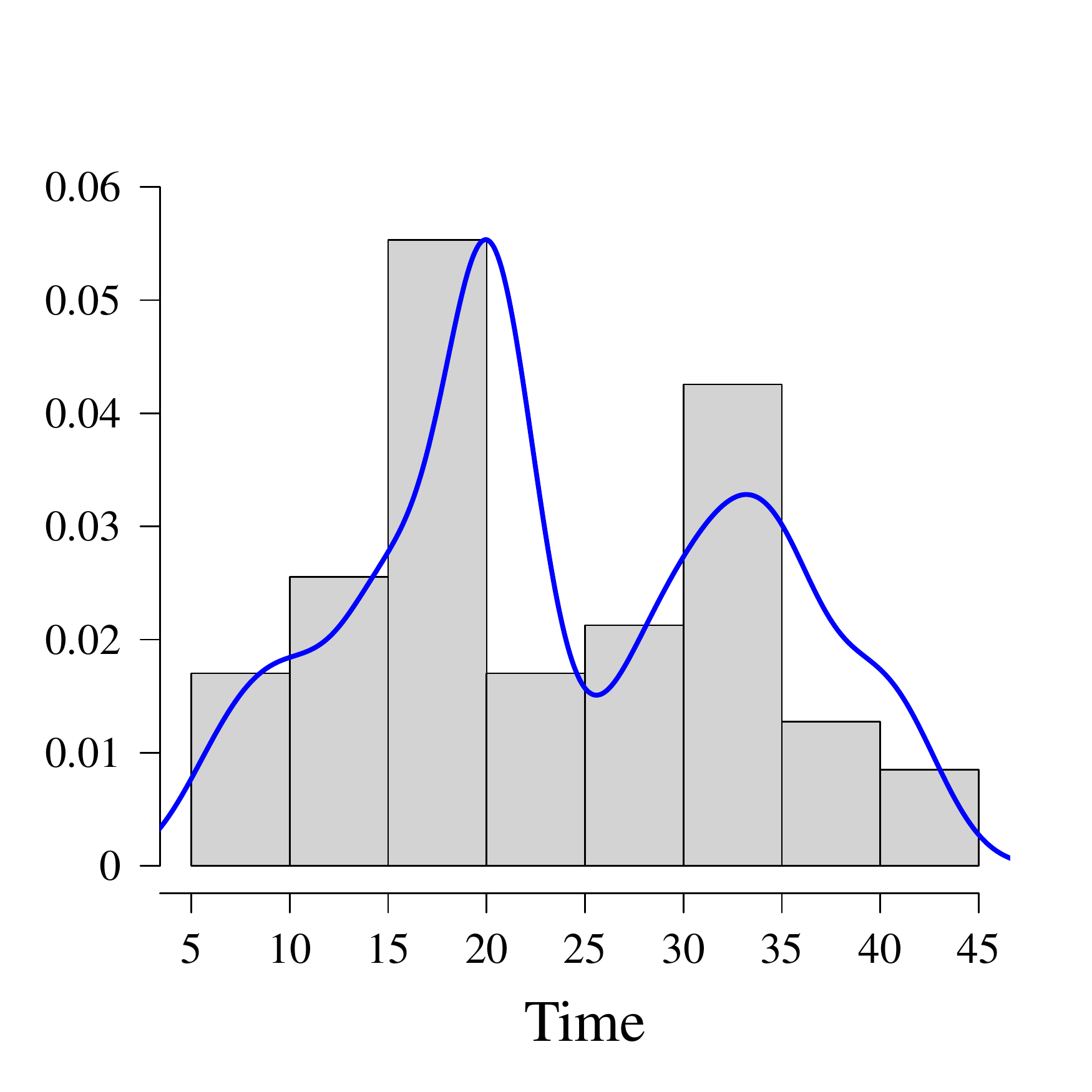} &
        \includegraphics[scale = 0.28]{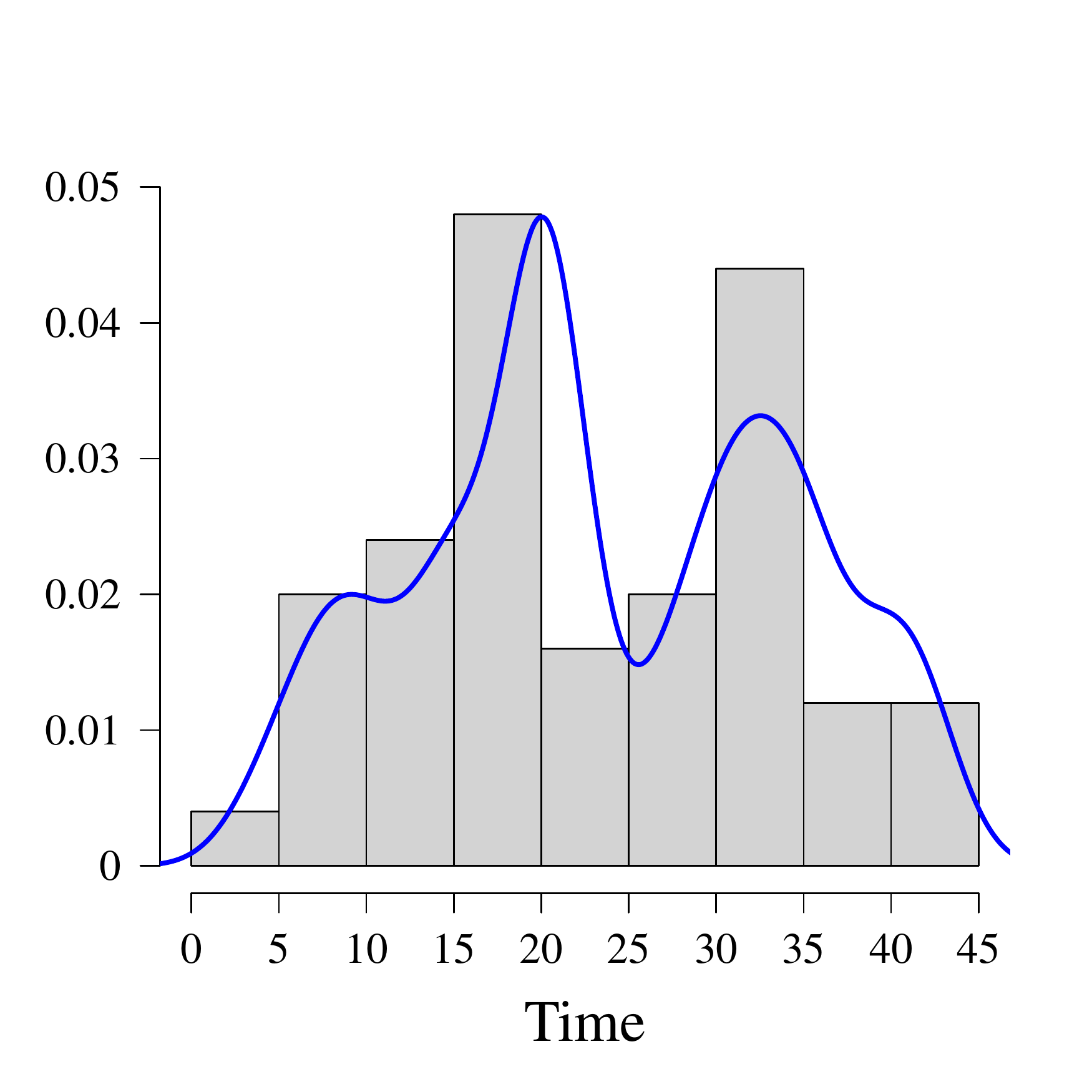} &
        \includegraphics[scale = 0.28]{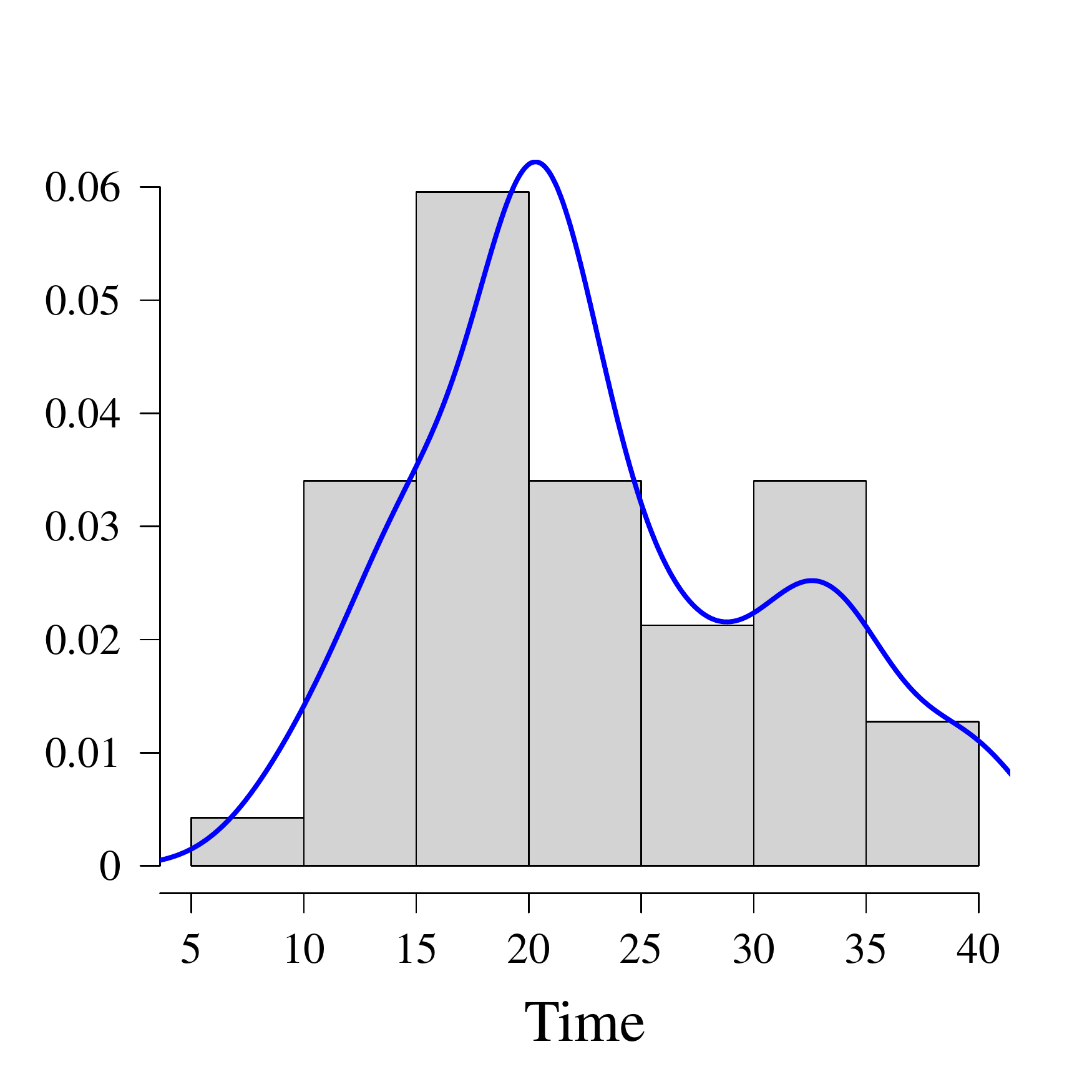} \\
        (a) minP-BH, $n_{s}=47$ & (b) LR-BH, $n_{s}=50$ & (c) CU1-BH, $n_{s}=47$
    \end{tabular}
    \caption{Changepoint locations from degree-based analysis of the MIT reality mining data.}
    \label{fig:mitdeg_hist}
\end{figure}
\begin{figure}[!ht]
    \begin{tabular}{cc}
        \includegraphics[scale = 0.45]{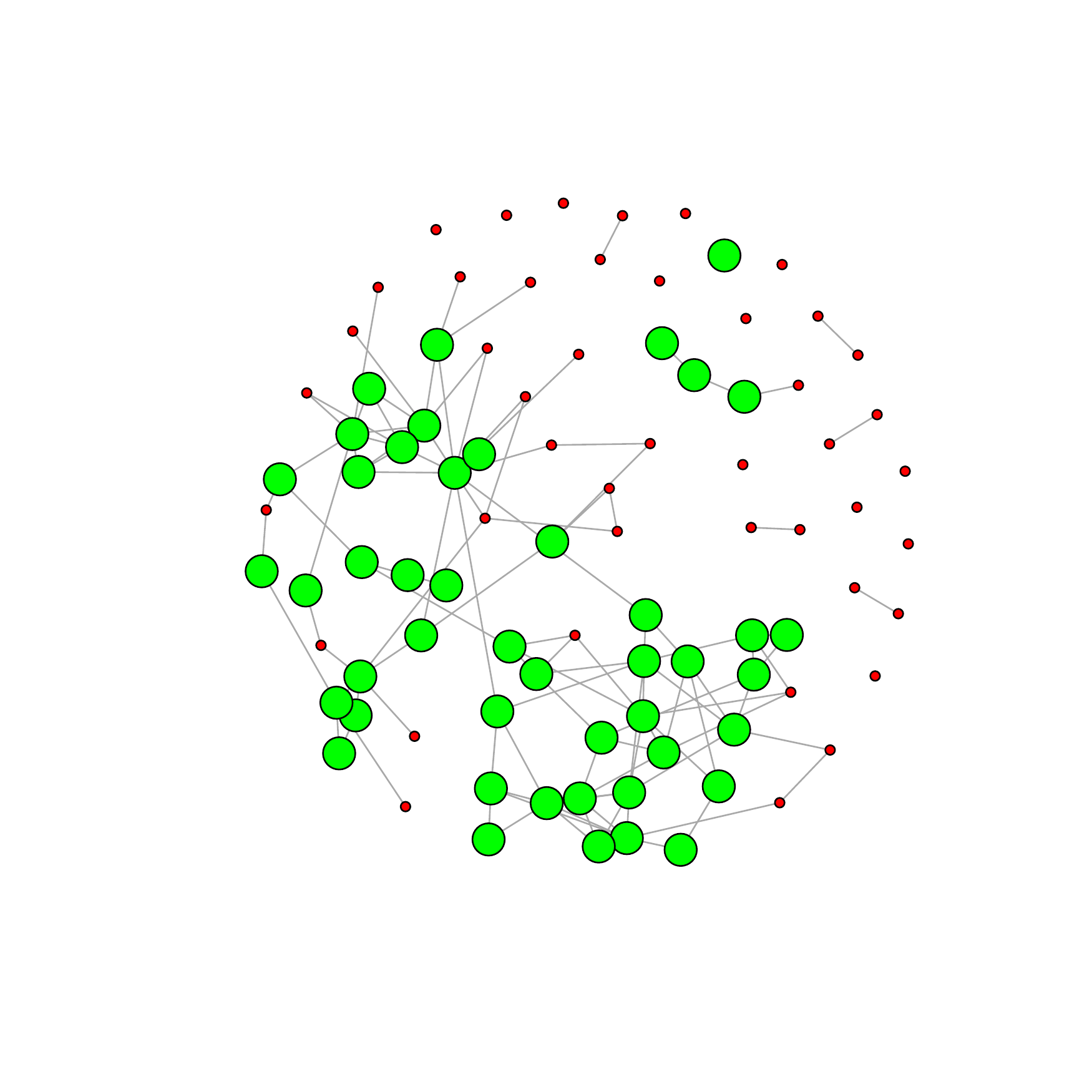} & 
        \includegraphics[scale = 0.45]{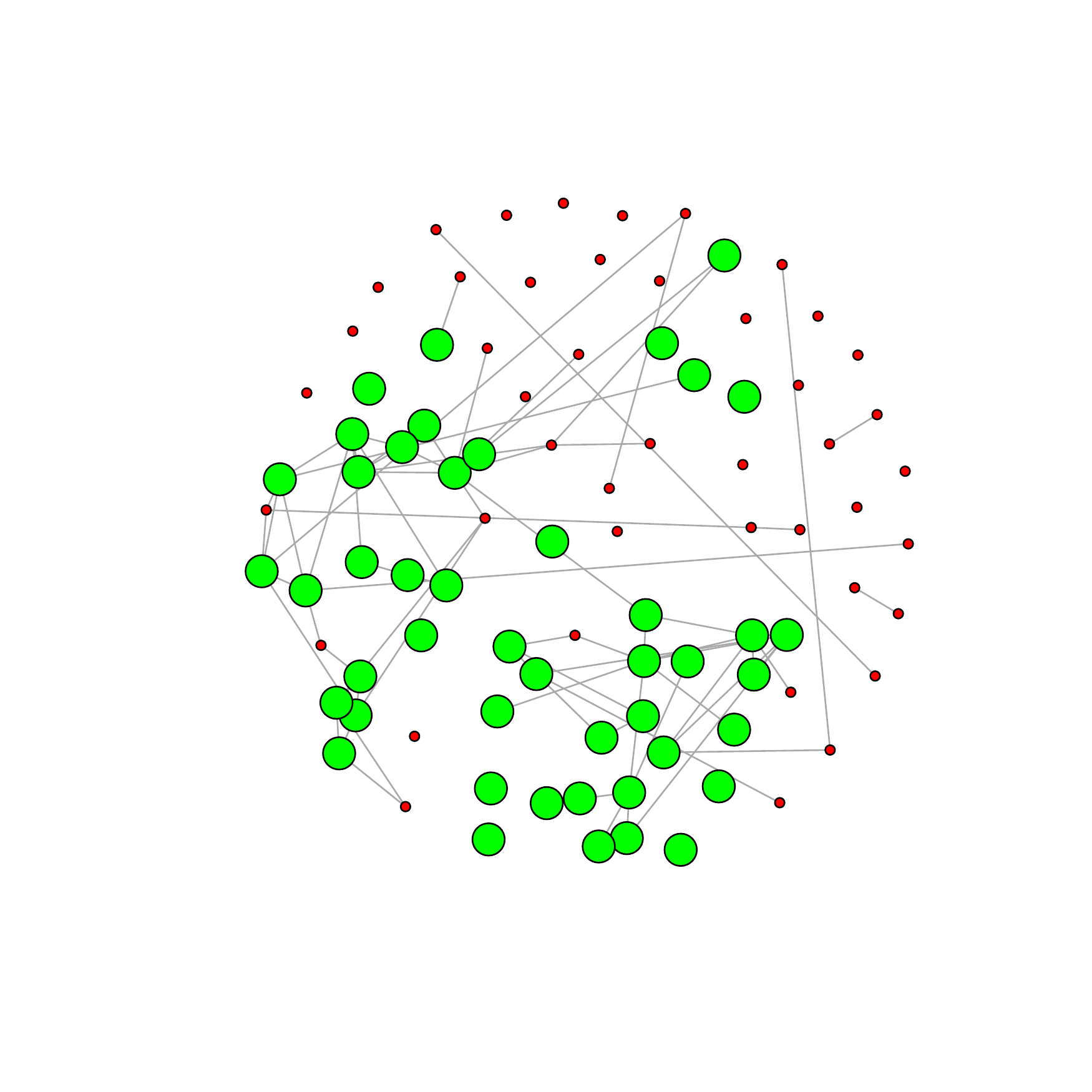} \\
        \quad(a) & \quad(b) 
    \end{tabular}
    \caption{Average networks (a) before and (b) after time-point 20 in the MIT reality mining data. Green nodes represent significant channels.}
    \label{fig:mit_graph}
\end{figure}

\section{Discussion}
\label{sec:discuss}
In this article, we have considered the problem of changepoint detection for binary and count data, and proposed exact tests that perform significantly better than Brownian-bridge based asymptotic tests in small samples. We have also considered multichannel data and used a multiple testing approach to test for changes in all channels simultaneously. This local approach outperforms the global approach of treating all channels together as a single object quite significantly in case of rare signals (i.e. when the number of channels with a changepoint is much smaller than the total number of channels).

Although the methods we propose are technically for single changes, they work quite well for multiple changepoints, especially when there is one large change. This is empirically demonstrated in the appendix.

\bibliographystyle{apalike}
\bibliography{bib}

\newpage
\appendix
\section{Appendix}

In this appendix, we evaluate the performances of the global and local tests conducted at level $\alpha=0.1$ for binary and count data. To construct the local tests, we implement the BH procedure and two more powerful adaptive methods for FDR control. The first one is the adaptive Benjamini-Hochberg (ABH) method where the number of true null hypotheses $m_0$ is estimated as $\widehat{m}_0$ using the proposal of \cite{hochberg1990more}, and this estimate is used in the BH method with $\alpha^{\ast}=\alpha m/\widehat{m}_0$. This approach gives better power than the BH method provided $m/\widehat{m}_0$ is significantly bigger than 1. The other method we use is the adaptive Storey-Taylor-Siegmund (STS) method proposed by \cite{storey2004strong}. 

We consider $m$ independent Bernoulli (Tables \ref{tab1:bern}-\ref{tab4:bern}) and Poisson (Tables \ref{tab1:pois}-\ref{tab4:pois}) time-series of size $T$. Among these $m$ channels, $n_{\text{cp}}$ channels contain changepoints. We compare the probabilities of global change detection (P(gCD)) of the global and local tests. For any local test, one of the tests among minP, LR and CU1 is used in each channel, and then one of the procedures among BH, ABH, and STS is applied to control FDR at level $\alpha=0.1$. Once channels are detected with changepoints by a local test, we provide a measure of the power of that local test in terms of its True Positive Rate (TPR) defined as $E[TP]/\max(1, n_{\text{cp}})$ where $TP$ is the number of channels with correctly detected changepoints. $P(gCD), TPR \text{ and } FDR$, presented in Tables \ref{tab1:bern}-\ref{tab4:pois}, are estimated based on 1000 Monte Carlo samples. 

\subsection{Studies with a single changepoint}
Tables \ref{tab1:bern}-\ref{tab3:bern} deal with Bernoulli time-series having a single changepoint at location $\tau$ in $n_{\text{cp}}$ channels. We make the following observations.
\begin{enumerate}
\item Power P(gCD) of the local tests are significantly higher than the global tests in setups where $n_{\text{cp}}$ is not large enough. This supports our claim that, compared to the global tests, local tests are more sensitive towards detecting a global change when signal is rare. If enough signal is present ($n_{\text{cp}}$ large), then global tests tend to produce similar power as the local tests. 

\item Under the sparse signal regime ($\pi_1 = 0.01, \pi_2 =0.1$), if the changepoint is closer to the boundary ($T=200, \tau=175$), then global tests fail miserably to detect a global change (Table \ref{tab2:bern}) whereas local tests are able to detect such a change with high probability even for small to moderate signal ($n_{\text{cp}}=10, 20$ for $m=200$, and $n_{\text{cp}}=20, 40$ for $m=1000$).

\item Performances of the three FDR controlling procedures are similar, although STS seems to yield slightly more power in most of the cases. 

\item gCU1 performs much better than gCU.5 when the changepoint occurs in the middle of the time-series, and gCU.5 significantly outperforms gCU1 when the changepoint is near the boundary. 

\item Both $P(gCD)$ and $TPR$ increase as $n_{\text{cp}}$ increases or $m$ decreases (see Tables \ref{tab1:bern}, \ref{tab2:bern} or Tables \ref{tab1:bern}, \ref{tab3:bern}). 

\item The CU1 local tests perform better than the minP and the LR local tests in most of the cases except for a few cases in Table \ref{tab2:bern} where the changepoint is near the boundary. 
\end{enumerate}

Tables \ref{tab1:pois}-\ref{tab3:pois} deal with Poisson time-series of small ($T=10$) and moderate ($T=50$) sizes having a single changepoint at $\tau$ in $n_{\text{cp}}$ channels. We observe the same phenomena mentioned in points 1, 3 and 4 above. If the channels have strong signal (e.g., $\lambda_1=0.3, \lambda_2=3$), then the local tests perform well even for small sample sizes (e.g., $T=10$) and very small $n_{\text{cp}}$ (e.g., $n_{\text{cp}}=2,6$). For cases with moderate sample size (e.g., $T=50$), global changes are detected with high probability even with weak signal (e.g., $\lambda_1=0.25, \lambda_2=0.75$). For settings with $T=50$, the global and local tests return higher values of $P(gCD)$ and $TPR$ when the changepoint occurs near the middle of the time-series. This simulation study also indicates that there is no uniformly best conditional exact test among minP, LR and CU1.

\subsection{Studies with multiple changepoints}
For a binary or count time-series with multiple changepoints, the exact conditional tests minP, LR and CU1 are theoretically valid level-$\alpha$ tests. We conduct comparative power analyses of these tests under the scenarios described in Tables~\ref{tab4:bern} and \ref{tab4:pois}. 

Table \ref{tab4:bern} deals with $m$ independent Bernoulli time-series of the form: $X_1,\dots,X_{\tau_1} \iid \bern(\pi_1)$, $X_{\tau_1+1},\dots,X_{\tau_2} \iid \bern(\pi_2)$, $X_{\tau_2+1},\dots,X_{\tau_3} \iid \bern(\pi_3)$, $X_{\tau_3+1},\dots,X_{T} \iid \bern(\pi_4)$ with $T=200$, $m=200$. $n_{\text{cp}}$ channels contain changepoints at $\tau_1 = 50, \tau_2 = 100$ and $\tau_3 = 150$. Interestingly, the LR test performs better than the minP and CU1 tests in all the scenarios considered here. Under the rare signal regime, the local tests produce much higher $P(gCD)$ than the global tests which supports the claims of Section~\ref{sec:discuss}. In the first setting ($\pi_1 =0.01, \pi_2 = 0.1, \pi_3= 0.2, \pi_4= 0.3$) in Table \ref{tab4:bern}, the Bernoulli probabilities increase steadily over time and strong signals are present in channels with changepoints. Thus $TPR$ and $P(gCD)$ are significantly high even for quite small $n_{\text{cp}}$. In the second setup ($\pi_1 =0.01, \pi_2 = 0.2, \pi_3= 0.1, \pi_4= 0.3$), the Bernoulli probabilities increase and decrease alternatively as time progresses. The local and the global tests experience in this setting a loss of power compared to the first setting. However, the local tests continue to uniformly (and significantly) dominate over the global tests. Note also that the CU1 local tests suffer the most whereas the LR local tests work the best. 

In Table \ref{tab4:pois}, we consider $m$ independent Poisson time-series of the form: $X_1,\dots,X_{\tau_1} \iid \pois(\lambda_1)$, $X_{\tau_1+1},\dots,X_{\tau_2} \iid \pois(\lambda_2)$, $X_{\tau_2+1},\dots,X_{\tau_3} \iid \pois(\lambda_3)$, $X_{\tau_3+1},\dots,X_{T} \iid \pois(\lambda_4)$ where $T=50$, $m=200$ and $n_{\text{cp}}$ channels contain changepoints at $\tau_1 = 10, \tau_2 = 20$ and $\tau_3 = 35$. Interestingly, the CU1 local tests outperform all other local tests for both the settings presented in Table \ref{tab4:pois}. The local and the global tests produce less power in the second setting ($\lambda_1 = 0.2$, $\lambda_2 = 0.8$, $\lambda_3 = 0.5$, $\lambda_4 = 1.4$) compared to the first setting ($\lambda_1 = 0.2$, $\lambda_2 = 0.5$, $\lambda_3 = 0.8$, $\lambda_4 = 1.4$). However, under both the settings, the local tests produce much higher power than the global tests when very few channels ($n_{\text{cp}}=2,4,8$) undergo changes, which again corroborates the claims of Section~\ref{sec:discuss}. 


\begin{table}
\caption{Comparison of gCU.5 and gCU1 with local tests obtained by combining the FDR controlling procedures BH, ABH and STS with minP, LR, and CU1 in $m$ independent Bernoulli series of the form: $X_1,\dots,X_{\tau} \iid \bern(\pi_1)$, $X_{\tau+1},\dots,X_{T} \iid \bern(\pi_2)$. Tests are conducted at level $\alpha=0.1$ and $n_{\text{cp}}$ channels (among $m$) contain a change at $\tau$.} 
\label{tab1:bern}
\begin{center}
\scalebox{.85}[0.85]{
\begin{tabular}{ |c c|c c| c c c|c c c|c c c|  }
\hline
\multicolumn{13}{|c|}{$m=200$, $T=50$, $\tau =25$, $\pi_1 = 0.01$, $\pi_2 = 0.30$  } \\
\hline \hline
             &  &  &  & minP & LR & CU1 & minP & LR & CU1     & minP & LR & CU1 \\
 $n_{\text{cp}}$ &   & gCU.5 & gCU1 & BH   & BH & BH   & ABH  & ABH & ABH & STS & STS & STS \\ [.5 ex]
\hline 
 & & & & & & & & & & & & \\ [-1.5ex]
0 & $P(gCD)$ & 0.114 & 0.082 & 0.086 & 0.108 & 0.09 & 0.086 & 0.108 & 0.09 & 0.09 & 0.11  & 0.092 \\ [.5 ex]
\hline   
 & & & & & & & & & & & & \\ [-1.5ex]
5 & $P(gCD)$ & 0.276 & 0.432 & 0.430 & 0.506 & 0.540 & 0.430 & 0.506 & 0.540 & 0.436 & 0.508 & 0.546\\ [.5 ex]
   & $TPR$     &           &            & 0.122 & 0.168 & 0.182 & 0.122 & 0.168 & 0.182 & 0.124 & 0.169 & 0.186 \\ [.5 ex]
   & $FDR$     &            &          & 0.091 & 0.092 & 0.088 & 0.091 & 0.094 & 0.089 & 0.091 & 0.094 & 0.089 \\ [.5 ex]
10 & $P(gCD)$ & 0.636 & 0.806  & 0.704 & 0.784 & 0.812 & 0.710 & 0.786 & 0.816 & 0.726 & 0.788 & 0.830 \\ [.5 ex]
   & $TPR$       &            &            & 0.182 & 0.238 & 0.265 & 0.185 & 0.239 & 0.270 & 0.191 & 0.242 & 0.280 \\ [.5 ex]
   & $FDR$     &              &            & 0.081 & 0.074 & 0.076 & 0.084 & 0.076 & 0.080 & 0.086 & 0.078 & 0.082 \\ [.5 ex]
15 & $P(gCD)$ & 0.912 & 0.974  & 0.854 & 0.914 & 0.924 & 0.856 & 0.916 & 0.928 & 0.868 & 0.920 & 0.930 \\ [.5 ex]
   & $TPR$       &           &             & 0.241 & 0.302 & 0.340 & 0.249 & 0.310 & 0.359 & 0.253 & 0.311 & 0.360 \\ [.5 ex]
   & $FDR$       &           &             & 0.089 & 0.090 & 0.091 & 0.093 & 0.093 & 0.096 & 0.095 & 0.094 & 0.094 \\ [.5 ex]
\hline
\hline
\multicolumn{13}{|c|}{$m=200$, $T=50$, $\tau = 40$, $\pi_1 = 0.01$, $\pi_2 = 0.30$ } \\
\hline 
\hline
 & & & & & & & & & & & & \\ [-1.5ex]
5 & $P(gCD)$ & 0.182 & 0.148 & 0.616 & 0.614 & 0.624 & 0.616 & 0.614 & 0.624 & 0.622 & 0.616 & 0.628 \\ [.5 ex]
   & $TPR$     &           &            & 0.193 & 0.194 & 0.219 & 0.193 & 0.194 & 0.220 & 0.195 & 0.195 & 0.221 \\ [.5 ex]
   & $FDR$     &           &            & 0.098 & 0.093 & 0.090 & 0.098 & 0.094 & 0.091 & 0.100 & 0.097 & 0.091 \\ [.5 ex]
10 & $P(gCD)$ & 0.458 & 0.178 & 0.832 & 0.828 & 0.852 & 0.834 & 0.828 & 0.852 & 0.838 & 0.832 & 0.854 \\ [.5 ex]
   & $TPR$       &            &           & 0.234 & 0.240 & 0.266 & 0.234 & 0.241 & 0.267 & 0.238 & 0.243 & 0.268 \\ [.5 ex]
   & $FDR$     &              &          & 0.084 & 0.083 & 0.080 & 0.088 & 0.084 & 0.081 & 0.088 & 0.086 & 0.082 \\ [.5 ex]
15 & $P(gCD)$ & 0.754 & 0.304 & 0.936 & 0.936 & 0.950 & 0.936 & 0.936 & 0.952 & 0.936 & 0.936 & 0.954 \\ [.5 ex]
   & $TPR$       &           &           & 0.308 & 0.311 & 0.329 & 0.309 & 0.312 & 0.333 & 0.312 & 0.314 & 0.337 \\ [.5 ex]
   & $FDR$       &          &          & 0.090 & 0.091 & 0.087 & 0.092 & 0.094 & 0.089 & 0.095 & 0.096 & 0.091 \\ [.5 ex]
\hline
\hline
\multicolumn{13}{|c|}{$m=200$, $T=200$, $\tau =100$, $\pi_1 = 0.01$, $\pi_2 = 0.10$  } \\
\hline 
\hline
 & & & & & & & & & & & & \\ [-1.5ex]
0 & $P(gCD)$ & 0.084 & 0.096 & 0.104 & 0.098 & 0.090 & 0.104 & 0.098 & 0.090 & 0.106 & 0.102 & 0.092 \\ [.5 ex]
\hline   
 & & & & & & & & & & & & \\ [-1.5ex] 
10 & $P(gCD)$ & 0.232 & 0.320 & 0.726 & 0.828 & 0.854 & 0.728 & 0.830 & 0.858 & 0.744 & 0.838 & 0.864 \\ [.5 ex]
   & $TPR$     &            &            & 0.189 & 0.258 & 0.287 & 0.192 & 0.260 & 0.293 & 0.199 & 0.266 & 0.300 \\ [.5 ex]
   & $FDR$     &           &            & 0.097 & 0.097 & 0.088 & 0.097 & 0.097 & 0.091 & 0.099 & 0.100 & 0.089 \\ [.5 ex]
20 & $P(gCD)$ & 0.454 & 0.676 & 0.964 & 0.984 & 0.992 & 0.964 & 0.984 & 0.992 & 0.966 & 0.986 & 0.992 \\ [.5 ex]
   & $TPR$       &            &           & 0.303 & 0.378 & 0.431 & 0.314 & 0.388 & 0.446 & 0.323 & 0.393 & 0.452 \\ [.5 ex]
   & $FDR$     &              &           & 0.093 & 0.091 & 0.083 & 0.097 & 0.095 & 0.090 & 0.101 & 0.099 & 0.094 \\ [.5 ex]
30 & $P(gCD)$ & 0.794 & 0.928 & 1 &  1 &  1 &  1  &  1 &  1 &  1 &  1 & 1 \\ [.5 ex]
   & $TPR$       &           &           & 0.365 & 0.433 & 0.517 & 0.383 & 0.452 & 0.546 & 0.395 & 0.459 & 0.557 \\ [.5 ex]
   & $FDR$       &           &           & 0.083 & 0.083 & 0.082 & 0.090 & 0.091 & 0.093 & 0.094 & 0.095 & 0.096 \\ [.5 ex]
\hline
\end{tabular} }
\end{center}
\end{table}

\begin{table}
\caption{(Cont'd) Comparison of gCU.5 and gCU1 with local tests obtained by combining the FDR controlling procedures BH, ABH and STS with minP, LR, and CU1 in $m$ independent Bernoulli series of the form: $X_1,\dots,X_{\tau} \iid \bern(\pi_1)$, $X_{\tau+1},\dots,X_{T} \iid \bern(\pi_2)$. Tests are conducted at level $\alpha=0.1$ and $n_{\text{cp}}$ channels (among $m$) contain a change at $\tau$.} 
\label{tab2:bern}
\begin{center}
\scalebox{0.85}[0.85]{
\begin{tabular}{ |c c|c c| c c c|c c c|c c c|  }
\hline
\multicolumn{13}{|c|}{$m=200$, $T=200$, $\tau =175$, $\pi_1 = 0.01$, $\pi_2 = 0.10$  } \\
\hline 
\hline
               &  &  &  & minP & LR & CU1 & minP & LR & CU1     & minP & LR & CU1 \\
 $n_{\text{cp}}$ &  & gCU.5 & gCU1 & BH   & BH & BH   & ABH  & ABH & ABH & STS & STS & STS \\ [.5 ex]
\hline 
 & & & & & & & & & & & & \\ [-1.5ex]
10 & $P(gCD)$ & 0.148   & 0.120 & 0.720 & 0.684 & 0.610 & 0.722 & 0.684 & 0.610 & 0.722 & 0.688 & 0.620\\ [.5 ex]
     & $TPR$     &              &          & 0.137      & 0.119  & 0.104 & 0.139 & 0.120 & 0.104 & 0.141 & 0.121 & 0.106 \\ [.5 ex]
     & $FDR$     &              &          & 0.104      & 0.102 & 0.098 & 0.104 & 0.104 & 0.099 & 0.107 & 0.105 & 0.102 \\ [.5 ex]
20 & $P(gCD)$ & 0.190    & 0.130  & 0.926 & 0.894 & 0.866 & 0.926 & 0.894 & 0.866 & 0.926 & 0.898 & 0.870 \\ [.5 ex]
     & $TPR$       &            &           & 0.176      & 0.157 & 0.138 & 0.180 & 0.159 & 0.140 & 0.184 & 0.164 & 0.144 \\ [.5 ex]
     & $FDR$     &              &           & 0.093      & 0.090 & 0.103 & 0.095 & 0.092 & 0.103 & 0.097 & 0.098 & 0.105 \\ [.5 ex]
30 & $P(gCD)$ & 0.445   & 0.142 & 0.976      & 0.966 & 0.960 & 0.976  & 0.966 & 0.960 & 0.976 & 0.966 & 0.960 \\ [.5 ex]
     & $TPR$       &           &          & 0.199     & 0.178 & 0.156 & 0.208       & 0.184 & 0.159 & 0.217 & 0.189 & 0.166 \\ [.5 ex]
     & $FDR$       &           &         & 0.081       & 0.080 & 0.083 & 0.084      & 0.085 & 0.088 & 0.089 & 0.088 & 0.092 \\ [.5 ex]
\hline
\hline
\multicolumn{13}{|c|}{$m=1000$, $T=200$, $\tau =100$, $\pi_1 = 0.01$, $\pi_2 = 0.10$  } \\
\hline \hline 
 & & & & & & & & & & & & \\ [-1.5ex]
0 & $P(gCD)$ & 0.110 & 0.095 & 0.122 & 0.122 & 0.106 & 0.122 & 0.122 & 0.106  & 0.122 & 0.122 & 0.106 \\ [.5 ex]
\hline   
 & & & & & & & & & & & & \\ [-1.5ex] 
20 & $P(gCD)$ & 0.192 & 0.330 & 0.674 & 0.782 & 0.794 & 0.674 & 0.782 & 0.794 & 0.674 & 0.782 & 0.794  \\ [.5 ex]
   & $TPR$     &            &           & 0.073 & 0.118 & 0.121 & 0.073 & 0.118 & 0.121  & 0.073 & 0.118 & 0.122   \\ [.5 ex]
   & $FDR$     &            &           & 0.109 & 0.110 & 0.101  & 0.109  & 0.110 & 0.101 & 0.109  & 0.110 & 0.101  \\ [.5 ex]
40 & $P(gCD)$ & 0.370 & 0.650 & 0.882 & 0.966 & 0.968 & 0.882 & 0.966 & 0.968  & 0.886 & 0.968 & 0.968      \\ [.5 ex]
   & $TPR$     &            &           & 0.110 & 0.179 & 0.198 & 0.111 & 0.180 & 0.200  & 0.114 & 0.184 & 0.205  \\ [.5 ex]
   & $FDR$     &           &            & 0.102 & 0.096 & 0.098  & 0.102 & 0.096 & 0.098 & 0.103 & 0.099 & 0.101  \\ [.5 ex]
60 & $P(gCD)$ & 0.700 & 0.925  & 0.970 & 0.996 & 0.992 & 0.970 & 0.996 & 0.992 & 0.970 & 0.996 & 0.994     \\ [.5 ex]
   & $TPR$     &           &            & 0.171 & 0.252 & 0.286 & 0.174 & 0.257 & 0.292 & 0.180 & 0.261 & 0.300 \\ [.5 ex]
   & $FDR$     &           &           & 0.091 & 0.090 & 0.091  & 0.093 & 0.092 & 0.093  & 0.097 & 0.094 & 0.096   \\ [.5 ex]
\hline
\hline
\multicolumn{13}{|c|}{$m=1000$, $T=200$, $\tau =175$, $\pi_1 = 0.01$, $\pi_2 = 0.10$  } \\
\hline 
\hline   
 & & & & & & & & & & & & \\ [-1.5ex] 
20 & $P(gCD)$ & 0.148 & 0.125 &  0.682 & 0.682 & 0.656 & 0.682 & 0.682 & 0.656 & 0.682 & 0.682 & 0.656 \\ [.5 ex]
   & $TPR$     &            &          & 0.074 & 0.072 & 0.062 & 0.074 & 0.072 & 0.062  & 0.074 & 0.072 & 0.062   \\ [.5 ex]
   & $FDR$     &            &          & 0.113 & 0.115 & 0.109  & 0.113 & 0.115 & 0.109  & 0.115 & 0.115 & 0.109  \\ [.5 ex]
40 & $P(gCD)$ & 0.256 & 0.160 & 0.912 & 0.900 & 0.872 & 0.912 & 0.900 & 0.872  & 0.912 & 0.900 & 0.872    \\ [.5 ex]
   & $TPR$     &            &          & 0.100 & 0.088 & 0.075 & 0.100 & 0.088 & 0.075   & 0.101 & 0.089 & 0.075   \\ [.5 ex]
   & $FDR$     &            &         & 0.106 & 0.100 & 0.094  & 0.106 & 0.100 & 0.094   & 0.107 & 0.102 & 0.096    \\ [.5 ex]
60 & $P(gCD)$ & 0.354 & 0.185 & 0.990 & 0.982 & 0.974 & 0.990 & 0.982 & 0.974 & 0.990 & 0.982 & 0.976   \\ [.5 ex]
   & $TPR$     &            &          & 0.129 & 0.112 & 0.091 & 0.130 & 0.112 & 0.091 & 0.132 & 0.113 & 0.093 \\ [.5 ex]
   & $FDR$     &            &         & 0.096 & 0.092 & 0.093  & 0.097 & 0.092 & 0.093  & 0.099 & 0.094 & 0.094   \\ [.5 ex]
\hline
\end{tabular} }
\end{center}
\end{table}

\begin{table}
\caption{(Cont'd) Comparison of gCU.5 and gCU1 with local tests obtained by combining the FDR controlling procedures BH, ABH and STS with minP, LR, and CU1 in $m$ independent Bernoulli series of the form: $X_1,\dots,X_{\tau} \iid \bern(\pi_1)$, $X_{\tau+1},\dots,X_{T} \iid \bern(\pi_2)$. Tests are conducted at level $\alpha=0.1$ and $n_{\text{cp}}$ channels (among $m$) contain a change at $\tau$.} 
\label{tab3:bern}
\begin{center}
\scalebox{0.85}[0.85]{
\begin{tabular}{ |c c|c c| c c c|c c c|c c c|  }
\hline
\multicolumn{13}{|c|}{$m=1000$, $T=50$, $\tau = 25$, $\pi_1 = 0.01$, $\pi_2 = 0.30$  } \\
\hline \hline
               &  &  & & minP & LR & CU1 & minP & LR & CU1     & minP & LR & CU1 \\
 $n_{\text{cp}}$ &  & gCU.5 & gCU1 & BH   & BH & BH   & ABH  & ABH & ABH & STS & STS & STS \\ [.5 ex]
\hline 
 & & & & & & & & & & & & \\ [-1.5ex]
0 & $P(gCD)$ & 0.096 & 0.11 &  0.086 & 0.086 & 0.080 & 0.086 & 0.086 & 0.080 & 0.086 & 0.086  & 0.080 \\ [.5 ex]
\hline   
 & & & & & & & & & & & & \\ [-1.5ex]
10 & $P(gCD)$ & 0.258  & 0.404 & 0.356 & 0.402 & 0.428 & 0.356 & 0.402 & 0.428 & 0.356 & 0.404 & 0.428 \\ [.5 ex]
   & $TPR$     &             &          & 0.046  & 0.055  & 0.065 & 0.046  & 0.055  & 0.065 & 0.046  & 0.055  & 0.065 \\ [.5 ex]
   & $FDR$     &             &          & 0.091  & 0.109 & 0.095  & 0.091  & 0.109 & 0.095   & 0.091  & 0.110 & 0.095  \\ [.5 ex]
20 & $P(gCD)$ & 0.566 &  0.792 & 0.590 & 0.650 & 0.722 & 0.590 & 0.650 & 0.722 & 0.590 & 0.650 & 0.722  \\ [.5 ex]
   & $TPR$       &           &          & 0.065 & 0.085 & 0.099 & 0.065 & 0.085 & 0.099 & 0.065 & 0.085 & 0.100 \\ [.5 ex]
   & $FDR$     &             &           & 0.098 & 0.094 & 0.089 & 0.098 & 0.094 & 0.089 & 0.098 & 0.094 & 0.088 \\ [.5 ex]
30 & $P(gCD)$ & 0.850 & 0.966  & 0.744 & 0.810 & 0.822 & 0.744 & 0.810 & 0.822  & 0.744 & 0.810 & 0.822  \\ [.5 ex]
   & $TPR$       &           &          & 0.081 & 0.108 & 0.122 & 0.081 & 0.109 & 0.122 & 0.081 & 0.108 & 0.124 \\ [.5 ex]
   & $FDR$       &           &          & 0.085  & 0.082 & 0.078  & 0.085  & 0.083 & 0.078 & 0.085 & 0.083 & 0.079 \\ [.5 ex]
\hline
\hline
\multicolumn{13}{|c|}{$m=1000$, $T=50$, $\tau =40$, $\pi_1 = 0.01$, $\pi_2 = 0.30$ } \\
\hline 
\hline  
 & & & & & & & & & & & & \\ [-1.5ex] 
10 & $P(gCD)$ & 0.182  & 0.132  & 0.494 & 0.474 & 0.614 & 0.494 & 0.474 & 0.614 & 0.494 & 0.474 & 0.614 \\ [.5 ex]
   & $TPR$     &              &          & 0.084 & 0.071 & 0.103 & 0.084 & 0.071 & 0.103 & 0.084 & 0.071 & 0.103  \\ [.5 ex]
   & $FDR$     &              &          & 0.090 & 0.109 & 0.098   & 0.090 & 0.109 & 0.098   & 0.090 & 0.109 & 0.098  \\ [.5 ex]
20 & $P(gCD)$ & 0.418  & 0.166  & 0.782 & 0.754 &0.862  & 0.782 & 0.754 & 0.862  & 0.782 & 0.754 & 0.862      \\ [.5 ex]
   & $TPR$     &             &          & 0.125 & 0.119 & 0.145 & 0.125 & 0.119 & 0.145  & 0.125 & 0.120 & 0.146 \\ [.5 ex]
   & $FDR$     &             &          & 0.098 & 0.092 & 0.102   & 0.098 & 0.092 & 0.102  & 0.098 & 0.092 & 0.102  \\ [.5 ex]
30 & $P(gCD)$ & 0.670 & 0.242  & 0.898 & 0.894 & 0.936 & 0.898 & 0.894 & 0.936  & 0.898 & 0.894 & 0.936     \\ [.5 ex]
   & $TPR$     &            &         & 0.147 & 0.140 & 0.172  & 0.147 & 0.140 & 0.172   & 0.148 & 0.140 & 0.173  \\ [.5 ex]
   & $FDR$     &            &         & 0.079 & 0.076 & 0.078   & 0.079 & 0.076 & 0.078     & 0.079 & 0.076 & 0.078     \\ [.5 ex]
\hline
\end{tabular} }
\end{center}
\end{table}

\begin{table}
\caption{[Multiple changepoints] Comparison of gCU.5 and gCU1 with local tests obtained by combining the FDR controlling procedures BH, ABH and STS with minP, LR, and CU1 in $m$ independent Bernoulli series of the form: $X_1,\dots,X_{\tau_1} \iid \bern(\pi_1)$, $X_{\tau_1+1},\dots,X_{\tau_2} \iid \bern(\pi_2)$, $X_{\tau_2+1},\dots,X_{\tau_3} \iid \bern(\pi_3)$ and $X_{\tau_3+1},\dots,X_{T} \iid \bern(\pi_4)$, where $T=200$ and $m=200$. Tests are conducted at level $\alpha=0.1$ and $n_{\text{cp}}$ channels (among $m$) contain 3 changes at locations $\tau_1 = 50, \tau_2 = 100$ and $\tau_3 = 150$.} 
\label{tab4:bern}
\begin{center}
\scalebox{0.85}[0.85]{
\begin{tabular}{ |c c|c c| c c c|c c c|c c c|  }
\hline
\multicolumn{13}{|c|}{ $\pi_1 = 0.01$, $\pi_2 = 0.1$, $\pi_3 = 0.2$, $\pi_4 = 0.3$ } \\
\hline 
\hline
               &  & & & minP & LR & CU1 & minP & LR & CU1     & minP & LR & CU1 \\
 $n_{\text{cp}}$ &  & gCU.5 & gCU1 & BH   & BH & BH   & ABH  & ABH & ABH & STS & STS & STS \\ [.5 ex]
\hline 
 & & & & & & & & & & & & \\ [-1.5ex]
0 & $P(gCD)$ & 0.084 & 0.096 & 0.104 & 0.098 & 0.090 & 0.104 & 0.098 & 0.090 & 0.106 & 0.102  & 0.092 \\ [.5 ex]
\hline   
 & & & & & & & & & & & & \\ [-1.5ex]
2 & $P(gCD)$ & 0.227   & 0.328 & 0.930 & 0.954 & 0.922 & 0.930 & 0.954 & 0.922 & 0.930 & 0.954 & 0.930 \\ [.5 ex]
   & $TPR$     &             &           & 0.761  & 0.811  & 0.764  & 0.761  & 0.811  & 0.764   & 0.763  & 0.812  & 0.776 \\ [.5 ex]
   & $FDR$     &             &          & 0.092  & 0.087 & 0.101   & 0.092  & 0.087 & 0.101   & 0.094  & 0.089 & 0.105   \\ [.5 ex]
4 & $P(gCD)$ & 0.556  &  0.660 & 1 & 1 & 0.992 & 1 & 1 & 0.992& 1 & 1 & 0.992 \\ [.5 ex]
   & $TPR$       &           &           & 0.832 & 0.861 & 0.832  & 0.832 & 0.863 & 0.833  & 0.836 & 0.865 & 0.837  \\ [.5 ex]
   & $FDR$     &             &          & 0.110 & 0.104 & 0.100  & 0.110 & 0.106 & 0.102  & 0.112 & 0.105 & 0.104 \\ [.5 ex]
6 & $P(gCD)$ & 0.826 & 0.892  & 1 & 1  & 1  & 1  & 1  & 1  & 1  & 1  & 1   \\ [.5 ex]
   & $TPR$       &          &           & 0.866 & 0.889 & 0.869 & 0.867 & 0.890 & 0.871 & 0.870 & 0.891 & 0.872 \\ [.5 ex]
   & $FDR$       &          &           & 0.102  & 0.096 & 0.099  & 0.103  & 0.098 & 0.103 & 0.104 & 0.098 & 0.104 \\ [.5 ex]
\hline
\hline
\multicolumn{13}{|c|}{ $\pi_1 = 0.01$, $\pi_2 = 0.2$, $\pi_3 = 0.1$, $\pi_4 = 0.3$ } \\
\hline 
\hline  
 & & & & & & & & & & & & \\ [-1.5ex] 
2 & $P(gCD)$ & 0.188   & 0.156 & 0.692 & 0.806 & 0.476 & 0.692 & 0.808 & 0.476 & 0.692 & 0.808 & 0.484 \\ [.5 ex]
   & $TPR$     &             &           & 0.475 & 0.602 & 0.281  & 0.475 & 0.605 & 0.282  & 0.479 & 0.605 & 0.288   \\ [.5 ex]
   & $FDR$     &             &           & 0.090 & 0.086 & 0.086    & 0.090 & 0.085 & 0.086    & 0.091 & 0.088 & 0.089   \\ [.5 ex]
4 & $P(gCD)$ & 0.402   & 0.254 & 0.910 & 0.978 & 0.692  & 0.910 & 0.978 & 0.692     & 0.910 & 0.978 & 0.698    \\ [.5 ex]
   & $TPR$     &             &           & 0.570 & 0.689 & 0.321  & 0.571 & 0.693 & 0.322   & 0.576 & 0.692 & 0.327 \\ [.5 ex]
   & $FDR$     &             &          & 0.109 & 0.104 & 0.099   & 0.109 & 0.104 & 0.103  & 0.110 & 0.104 & 0.102  \\ [.5 ex]
6 & $P(gCD)$ & 0.632  & 0.352 & 0.982 & 0.994 & 0.866 & 0.982 & 0.996 & 0.866  & 0.984 & 0.994 & 0.872     \\ [.5 ex]
   & $TPR$     &            &           & 0.650 & 0.751 & 0.384  & 0.651 & 0.754 & 0.387    & 0.655 & 0.756 & 0.391    \\ [.5 ex]
   & $FDR$     &            &          & 0.102 & 0.100 & 0.100    & 0.105 & 0.102 & 0.101      & 0.106 & 0.103 & 0.105       \\ [.5 ex]
\hline
\end{tabular} }
\end{center}
\end{table}


\begin{table}
\caption{Comparison of gCU.5 and gCU1 with local tests obtained by combining the FDR controlling procedures BH, ABH and STS with minP, LR, and CU1 in $m$ independent Poisson series of the form: $X_1,\dots,X_{\tau} \iid \pois(\lambda_1)$,\, $X_{\tau+1},\dots,X_{T} \iid \pois(\lambda_2)$. Tests are conducted at level $\alpha=0.1$ and $n_{\text{cp}}$ channels (among $m$) contain a change at $\tau$.} 
\label{tab1:pois}
\begin{center}
\scalebox{0.85}[0.85]{
\begin{tabular}{ |c c|c c| c c c|c c c|c c c|  }
\hline
\multicolumn{13}{|c|}{$m=1000$, $T=50$, $\tau = 25$, $\lambda_1 = 0.15$, $\lambda_2 = 0.75$  } \\
\hline 
\hline
             &  &  &  & minP & LR & CU1 & minP & LR & CU1     & minP & LR & CU1 \\
 $n_{\text{cp}}$ &   & gCU.5 & gCU1 & BH   & BH & BH   & ABH  & ABH & ABH & STS & STS & STS \\ [.5 ex]
\hline 
 & & & & & & & & & & & & \\ [-1.5ex]
0 & $P(gCD)$ & 0.095 & 0.115 & 0.11 & 0.115 & 0.135 & 0.11 & 0.115 & 0.135 & 0.11 & 0.115 & 0.135 \\ [.5 ex]
\hline   
 & & & & & & & & & & & & \\ [-1.5ex]
10 & $P(gCD)$ & 0.225 & 0.325 & 0.830 & 0.900 & 0.890 & 0.830 & 0.900 & 0.890 & 0.835 & 0.900 & 0.890  \\ [.5 ex]
   & $TPR$     &           &            & 0.197 & 0.242 & 0.283 & 0.197 & 0.242 & 0.283 & 0.199 & 0.242 & 0.283 \\ [.5 ex]
   & $FDR$     &            &          & 0.113 & 0.102 & 0.099 & 0.113 & 0.102 & 0.099 & 0.114 & 0.102 & 0.099 \\ [.5 ex]
20 & $P(gCD)$ & 0.395 & 0.670  & 0.980 & 0.985 & 0.995 & 0.980 & 0.985 & 0.995 &  0.980 & 0.985 & 0.995 \\ [.5 ex]
   & $TPR$       &            &            & 0.260 & 0.277 & 0.350 & 0.261 & 0.277 & 0.350 & 0.262 & 0.278 & 0.352 \\ [.5 ex]
   & $FDR$     &              &            & 0.095 & 0.101 & 0.098 & 0.095 & 0.101 & 0.098 & 0.095 & 0.101 & 0.099 \\ [.5 ex]
30 & $P(gCD)$ & 0.655 & 0.905  & 0.995 & 0.995 & 0.995 & 0.995 & 0.9995 & 0.995 & 0.995 & 0.9995 & 0.995 \\ [.5 ex]
   & $TPR$       &           &             & 0.298 & 0.320 & 0.400 & 0.298 & 0.321 & 0.401 & 0.303 & 0.322 & 0.402\\ [.5 ex]
   & $FDR$       &           &             & 0.085 & 0.095 & 0.093 & 0.086 & 0.095 & 0.094 & 0.091 & 0.099 & 0.095 \\ [.5 ex]
\hline
\hline
\multicolumn{13}{|c|}{$m=1000$, $T=50$, $\tau = 40$, $\lambda_1 = 0.15$, $\lambda_2 = 0.75$ } \\
\hline 
\hline 
 & & & & & & & & & & & & \\ [-1.5ex]
10 & $P(gCD)$ & 0.175 & 0.125 & 0.710 & 0.635 & 0.695 & 0.710 & 0.635 & 0.695 &  0.710 & 0.635 & 0.700 \\ [.5 ex]
   & $TPR$     &           &            & 0.147 & 0.113 & 0.132 & 0.147 & 0.113 & 0.132 & 0.147 & 0.113 & 0.134 \\ [.5 ex]
   & $FDR$     &           &            & 0.100 & 0.117 & 0.113 & 0.100 & 0.117 & 0.113  & 0.100 & 0.117 & 0.113  \\ [.5 ex]
20 & $P(gCD)$ & 0.250 & 0.165 & 0.920 & 0.900 & 0.920 & 0.920 & 0.900 & 0.920 & 0.920 & 0.900 & 0.920 \\ [.5 ex]
   & $TPR$       &            &           & 0.192 & 0.158 & 0.173 & 0.192 & 0.158 & 0.173 & 0.193 & 0.158 & 0.173 \\ [.5 ex]
   & $FDR$     &              &          & 0.078 & 0.090 & 0.095 & 0.080 & 0.090 & 0.095 & 0.080 & 0.090 & 0.095 \\ [.5 ex]
30 & $P(gCD)$ & 0.425 & 0.220 & 0.995 & 0.980 & 0.980 & 0.995 & 0.980 & 0.980 & 0.995 & 0.980 & 0.980 \\ [.5 ex]
   & $TPR$       &           &           & 0.234 & 0.184 & 0.216 & 0.234 & 0.184 & 0.216 & 0.236 & 0.185 & 0.217 \\ [.5 ex]
   & $FDR$       &          &          & 0.091 & 0.098 & 0.097 & 0.092 & 0.098 & 0.097 & 0.091 & 0.097 & 0.098 \\ [.5 ex]
\hline
\hline
\multicolumn{13}{|c|}{$m=200$, $T=50$, $\tau = 25$, $\lambda_1 = 0.25$, $\lambda_2 = 0.75$  } \\
\hline 
\hline  
 & & & & & & & & & & & & \\ [-1.5ex]
0 & $P(gCD)$ & 0.070 & 0.120 & 0.100 & 0.100 & 0.100 & 0.100 & 0.100 & 0.100 & 0.100 & 0.100 & 0.100 \\ [.5 ex]
\hline    
 & & & & & & & & & & & & \\ [-1.5ex] 
5 & $P(gCD)$ & 0.150 & 0.250 & 0.420 & 0.460 & 0.540 & 0.420 & 0.460 & 0.540 & 0.430 & 0.460 & 0.540 \\ [.5 ex]
   & $TPR$     &            &            & 0.112 & 0.122 & 0.152 & 0.112 & 0.122 & 0.152 & 0.114 & 0.122 & 0.154 \\ [.5 ex]
   & $FDR$     &           &            & 0.059 & 0.063 & 0.095 & 0.059 & 0.063 & 0.095 & 0.064 & 0.063 & 0.095 \\ [.5 ex]
10 & $P(gCD)$ & 0.260 & 0.640 & 0.680 & 0.690 & 0.790 & 0.680 & 0.690 & 0.790 & 0.690 & 0.690 & 0.800 \\ [.5 ex]
   & $TPR$       &            &           & 0.133 & 0.148 & 0.198 & 0.133 & 0.149 & 0.199 & 0.140 & 0.150 & 0.206 \\ [.5 ex]
   & $FDR$     &              &           & 0.092 & 0.100 & 0.094 & 0.092 & 0.095 & 0.096  & 0.099 & 0.100 & 0.106  \\ [.5 ex]
20 & $P(gCD)$ & 0.780 & 0.940 & 0.940 & 0.940 & 0.980 & 0.940 & 0.940 & 0.980 & 0.940 & 0.940 & 0.980 \\ [.5 ex]
   & $TPR$       &           &           & 0.187 & 0.190 & 0.261 & 0.189 & 0.192 & 0.270 & 0.196 & 0.196 & 0.276 \\ [.5 ex]
   & $FDR$       &           &           & 0.107 & 0.105 & 0.097 & 0.109 & 0.106 & 0.101 & 0.119 & 0.112 & 0.102 \\ [.5 ex]
\hline
\end{tabular} }
\end{center}
\end{table}

\begin{table}
\caption{(Cont'd) Comparison of gCU.5 and gCU1 with local tests obtained by combining the FDR controlling procedures BH, ABH and STS with minP, LR, and CU1 in $m$ independent Poisson series of the form: $X_1,\dots,X_{\tau} \iid \pois(\lambda_1)$, $X_{\tau+1},\dots,X_{T} \iid \pois(\lambda_2)$. Tests are conducted at level $\alpha=0.1$ and $n_{\text{cp}}$ channels (among $m$) contain a change at $\tau$.} 
\label{tab2:pois}
\begin{center}
\scalebox{0.85}[0.85]{
\begin{tabular}{ |c c|c c| c c c|c c c|c c c|  }
\hline
\multicolumn{13}{|c|}{$m=200$, $T=50$, $\tau = 40$, $\lambda_1 = 0.25$, $\lambda_2 = 0.75$  } \\
\hline \hline
             &  &  &  & minP & LR & CU1 & minP & LR & CU1     & minP & LR & CU1 \\
 $n_{\text{cp}}$ &   & gCU.5 & gCU1 & BH   & BH & BH   & ABH  & ABH & ABH & STS & STS & STS \\ [.5 ex]
\hline 
 & & & & & & & & & & & & \\ [-1.5ex]
10 & $P(gCD)$ & 0.200 & 0.190 & 0.510 & 0.450 & 0.440 &  0.510 & 0.470 & 0.440 & 0.520 & 0.460 & 0.450 \\ [.5 ex]
   & $TPR$     &           &            & 0.082 & 0.063 & 0.067 & 0.082 & 0.065 & 0.067 & 0.087 & 0.064 & 0.069 \\ [.5 ex]
   & $FDR$     &            &          & 0.089 & 0.083 & 0.087 & 0.089 & 0.089 & 0.087 & 0.096 & 0.089 & 0.097 \\ [.5 ex]
20 & $P(gCD)$ & 0.400 & 0.210  & 0.810 & 0.700 & 0.760 & 0.810 & 0.700 & 0.760  &  0.820 & 0.710 & 0.790  \\ [.5 ex]
   & $TPR$       &            &            & 0.128 & 0.094 & 0.101 & 0.133 & 0.096 & 0.103 & 0.136 & 0.096 & 0.110 \\ [.5 ex]
   & $FDR$     &              &            & 0.117 & 0.083 & 0.121 & 0.114 & 0.086 & 0.121 & 0.122 & 0.087 & 0.131 \\ [.5 ex]
30 & $P(gCD)$ & 0.720 & 0.290  & 0.910 & 0.880 & 0.910 & 0.910 & 0.880 & 0.910 & 0.920 & 0.880 & 0.920 \\ [.5 ex]
   & $TPR$       &           &             & 0.147 & 0.109 & 0.123 & 0.151 & 0.111 & 0.124 & 0.158 & 0.117 & 0.131 \\ [.5 ex]
   & $FDR$       &           &             & 0.076 & 0.068 & 0.078 & 0.076 & 0.068 & 0.077 & 0.082 & 0.079 & 0.087 \\ [.5 ex]
\hline
\hline
\multicolumn{13}{|c|}{$m=1000$, $T=10$, $\tau = 5$, $\lambda_1 = 0.30$, $\lambda_2 = 3.5$ } \\
\hline 
\hline  
 & & & & & & & & & & & & \\ [-1.5ex]
 0 & $P(gCD)$ & 0.090 & 0.040 & 0.070 & 0.060 & 0.090 & 0.070 & 0.060 & 0.090 & 0.070 & 0.060 & 0.090 \\ [.5 ex]
\hline   
 & & & & & & & & & & & & \\ [-1.5ex]
2 & $P(gCD)$ & 0.210 & 0.380 & 0.630 & 0.600 & 0.630 & 0.630 & 0.600 & 0.630 &  0.630 & 0.600 & 0.630 \\ [.5 ex]
   & $TPR$     &           &            & 0.410 & 0.395 & 0.420 & 0.410 & 0.395 & 0.420  & 0.410 & 0.395 & 0.420  \\ [.5 ex]
   & $FDR$     &           &            & 0.047 & 0.032 & 0.071 &  0.047 & 0.032 & 0.071  &  0.047 & 0.032 & 0.071   \\ [.5 ex]
4 & $P(gCD)$ & 0.450 & 0.700 & 0.800 & 0.780 & 0.820 & 0.800 & 0.780 & 0.820 & 0.800 & 0.780 & 0.820 \\ [.5 ex]
   & $TPR$       &            &           & 0.400 & 0.405 & 0.428 & 0.400 & 0.405 & 0.428 & 0.400 & 0.405 & 0.428 \\ [.5 ex]
   & $FDR$     &              &          & 0.049 & 0.038 & 0.045 & 0.049 & 0.038 & 0.045 & 0.049 & 0.038 & 0.045 \\ [.5 ex]
10 & $P(gCD)$ & 0.940 & 0.990 & 0.990 & 0.980 & 0.990 & 0.990 & 0.980 & 0.990 & 0.990 & 0.980 & 0.990 \\ [.5 ex]
   & $TPR$       &           &           & 0.474 & 0.533 & 0.560 & 0.474 & 0.533 & 0.560 & 0.474 & 0.533 & 0.560 \\ [.5 ex]
   & $FDR$       &          &          & 0.078 & 0.040 & 0.077 & 0.078 & 0.040 & 0.077 & 0.078 & 0.040 & 0.077 \\ [.5 ex]
\hline
\hline
\multicolumn{13}{|c|}{$m=1000$, $T=10$, $\tau = 7$, $\lambda_1 = 0.30$, $\lambda_2 = 3.5$  } \\
\hline 
\hline  
 & & & & & & & & & & & & \\ [-1.5ex] 
2 & $P(gCD)$ & 0.160 & 0.170 & 0.670 & 0.600 & 0.690 & 0.670 & 0.600 & 0.690 & 0.670 & 0.600 & 0.690 \\ [.5 ex]
   & $TPR$     &            &            & 0.440 & 0.370 & 0.475 & 0.440 & 0.370 & 0.475  & 0.440 & 0.370 & 0.475  \\ [.5 ex]
   & $FDR$     &           &            & 0.072 & 0.027 & 0.065 & 0.072 & 0.027 & 0.065 & 0.072 & 0.027 & 0.065 \\ [.5 ex]
4 & $P(gCD)$ & 0.490 & 0.310 & 0.840 & 0.760 & 0.870 & 0.840 & 0.760 & 0.870 & 0.840 & 0.760 & 0.870 \\ [.5 ex]
   & $TPR$       &            &           & 0.463 & 0.375 & 0.515 & 0.463 & 0.375 & 0.515 & 0.463 & 0.375 & 0.515 \\ [.5 ex]
   & $FDR$     &              &           & 0.047 & 0.034 & 0.050 & 0.047 & 0.034 & 0.050  & 0.047 & 0.034 & 0.050  \\ [.5 ex]
10 & $P(gCD)$ & 0.940 & 0.510 & 1 & 1 & 1 & 1 & 1 & 1 & 1 & 1 & 1 \\ [.5 ex]
   & $TPR$       &           &           & 0.557 & 0.511 & 0.603 & 0.557 & 0.511 & 0.603 & 0.557 & 0.511 & 0.603 \\ [.5 ex]
   & $FDR$       &           &           & 0.073 & 0.042 & 0.078 & 0.073 & 0.042 & 0.078 & 0.073 & 0.042 & 0.078 \\ [.5 ex]
\hline
\end{tabular} }
\end{center}
\end{table}

\begin{table}
\caption{(Cont'd) Comparison of gCU.5 and gCU1 with local tests obtained by combining the FDR controlling procedures BH, ABH and STS with minP, LR, and CU1 in $m$ independent Poisson series of the form: $X_1,\dots,X_{\tau} \iid \pois(\lambda_1)$,\, $X_{\tau+1},\dots,X_{T} \iid \pois(\lambda_2)$. Tests are conducted at level $\alpha=0.1$ and $n_{\text{cp}}$ channels (among $m$) contain a change at $\tau$.} 
\label{tab3:pois}
\begin{center}
\scalebox{0.85}[0.85]{
\begin{tabular}{ |c c|c c| c c c|c c c|c c c|  }
\hline
\multicolumn{13}{|c|}{$m=200$, $T=10$, $\tau = 5$, $\lambda_1 = 0.3$, $\lambda_2 = 3$ } \\
\hline \hline
             &  &  &  & minP & LR & CU1 & minP & LR & CU1     & minP & LR & CU1 \\
 $n_{\text{cp}}$ &   & gCU.5 & gCU1 & BH   & BH & BH   & ABH  & ABH & ABH & STS & STS & STS \\ [.5 ex]
\hline 
 & & & & & & & & & & & & \\ [-1.5ex]
0 & $P(gCD)$ & 0.096 & 0.101 & 0.073 & 0.067 & 0.083 & 0.073 & 0.067 & 0.083 & 0.074 & 0.067 & 0.083 \\ [.5 ex]
\hline   
 & & & & & & & & & & & & \\ [-1.5ex]
1 & $P(gCD)$ & 0.140 & 0.220 & 0.340 & 0.410 & 0.380 &  0.340 & 0.410 & 0.380 & 0.340 & 0.410 & 0.380 \\ [.5 ex]
   & $TPR$     &           &            & 0.290 & 0.400 & 0.350 & 0.290 & 0.400 & 0.350 & 0.290 & 0.400 & 0.350 \\ [.5 ex]
   & $FDR$     &            &          & 0.088 & 0.037 & 0.055 & 0.088 & 0.037 & 0.055  & 0.088 & 0.037 & 0.055  \\ [.5 ex]
2 & $P(gCD)$ & 0.360 & 0.500  & 0.520 & 0.590 & 0.580 & 0.520 & 0.590 & 0.580  &  0.520 & 0.590 & 0.580  \\ [.5 ex]
   & $TPR$       &            &            & 0.325 & 0.370 & 0.360 & 0.325 & 0.370 & 0.360 & 0.330 & 0.370 & 0.360 \\ [.5 ex]
   & $FDR$     &              &            & 0.045 & 0.032 & 0.032 & 0.045 & 0.032 & 0.033 & 0.045 & 0.032 & 0.032 \\ [.5 ex]
6 & $P(gCD)$ & 0.900 & 0.975  & 0.860 & 0.970 & 0.960 & 0.860 & 0.970 & 0.960  & 0.860 & 0.970 & 0.960  \\ [.5 ex]
   & $TPR$       &           &             & 0.407 & 0.492 & 0.473 & 0.407 & 0.500 & 0.473  & 0.407 & 0.492 & 0.473  \\ [.5 ex]
   & $FDR$       &           &             & 0.053 & 0.068 & 0.062 & 0.055 & 0.068 & 0.062 & 0.053 & 0.068 & 0.062 \\ [.5 ex]
\hline
\hline
\multicolumn{13}{|c|}{ $m=200$, $T=10$, $\tau = 7$, $\lambda_1 = 0.3$, $\lambda_2 = 3$ } \\
\hline 
\hline
 & & & & & & & & & & & & \\ [-1.5ex]
1 & $P(gCD)$ & 0.110 & 0.120 & 0.440 & 0.430 & 0.480 & 0.440 & 0.430 & 0.480  &  0.440 & 0.430 & 0.480  \\ [.5 ex]
   & $TPR$     &           &            & 0.390 & 0.400 & 0.430 & 0.390 & 0.400 & 0.430  & 0.390 & 0.400 & 0.430  \\ [.5 ex]
   & $FDR$     &           &            & 0.103 & 0.052 & 0.070 & 0.103 & 0.052 & 0.070  &  0.103 & 0.052 & 0.070  \\ [.5 ex]
2 & $P(gCD)$ & 0.350 & 0.210 & 0.590 & 0.540 & 0.710 & 0.590 & 0.540 & 0.710 & 0.590 & 0.540 & 0.710 \\ [.5 ex]
   & $TPR$       &            &           & 0.395 & 0.320 & 0.460 & 0.395 & 0.320 & 0.460 & 0.395 & 0.320 & 0.460 \\ [.5 ex]
   & $FDR$     &              &          & 0.068 & 0.033 & 0.042 & 0.068 & 0.033 & 0.043 & 0.068 & 0.033 & 0.042 \\ [.5 ex]
6 & $P(gCD)$ & 0.760 & 0.570 & 0.950 & 0.960 & 0.980 &  0.950 & 0.960 & 0.980 &  0.950 & 0.960 & 0.980 \\ [.5 ex]
   & $TPR$       &           &           & 0.517 & 0.470 & 0.538 & 0.517 & 0.472 & 0.538 & 0.517 & 0.470 & 0.538 \\ [.5 ex]
   & $FDR$       &          &          & 0.061 & 0.067 & 0.067 & 0.061 & 0.069 & 0.067 & 0.061 & 0.067 & 0.067 \\ [.5 ex]
\hline
\end{tabular} }
\end{center}
\end{table}

\begin{table}
\caption{[Multiple changepoints] Comparison of gCU.5 and gCU1 with local tests obtained by combining the FDR controlling procedures BH, ABH and STS with minP, LR, and CU1 in $m$ independent Poisson series of the form: $X_1,\dots,X_{\tau_1} \iid \pois(\lambda_1)$, $X_{\tau_1+1},\dots,X_{\tau_2} \iid \pois(\lambda_2)$, $X_{\tau_2+1},\dots,X_{\tau_3} \iid \pois(\lambda_3)$ and $X_{\tau_3+1},\dots,X_{T} \iid \pois(\lambda_4)$ where $T=50$ and $m=200$. Tests are conducted at level $\alpha=0.1$ and $n_{\text{cp}}$ channels (among $m$) contain 3 changes at locations $\tau_1 = 10, \tau_2 = 20$ and $\tau_3 = 35$.} 
\label{tab4:pois}
\begin{center}
\scalebox{0.85}[0.85]{
\begin{tabular}{ |c c|c c| c c c|c c c|c c c|  }
\hline
\multicolumn{13}{|c|}{ $\lambda_1 = 0.2$, $\lambda_2 = 0.5$, $\lambda_3 = 0.8$, $\lambda_4 = 1.4$ } \\
\hline 
\hline
               &  & & & minP & LR & CU1 & minP & LR & CU1     & minP & LR & CU1 \\
 $n_{\text{cp}}$ &  & gCU.5 & gCU1 & BH   & BH & BH   & ABH  & ABH & ABH & STS & STS & STS \\ [.5 ex]
\hline 
 & & & & & & & & & & & & \\ [-1.5ex]
0 & $P(gCD)$ & 0.102 & 0.098 & 0.118 & 0.116 & 0.114 &  0.118 & 0.116 & 0.114 &  0.122 & 0.116 & 0.114 \\ [.5 ex]
\hline   
 & & & & & & & & & & & & \\ [-1.5ex]
2 & $P(gCD)$ & 0.210   & 0.266 & 0.598 & 0.594 & 0.614 & 0.598 & 0.594 & 0.614 & 0.602 & 0.596 & 0.614 \\ [.5 ex]
   & $TPR$     &             &           & 0.364  & 0.369  & 0.382  & 0.364  & 0.369  & 0.382  &  0.365  & 0.371  & 0.383  \\ [.5 ex]
   & $FDR$     &             &          & 0.100  & 0.095 & 0.095   & 0.100  & 0.095 & 0.095  & 0.105  & 0.097 & 0.097   \\ [.5 ex]
4 & $P(gCD)$ & 0.344  &  0.502 & 0.816 & 0.826 & 0.838 & 0.816 & 0.826 & 0.838 & 0.818 & 0.830 & 0.846 \\ [.5 ex]
   & $TPR$       &           &           & 0.429 & 0.435 & 0.460  & 0.430 & 0.435 & 0.460  & 0.431 & 0.439 & 0.468  \\ [.5 ex]
   & $FDR$     &             &          & 0.092 & 0.099 & 0.087  & 0.092 & 0.099 & 0.088  & 0.097 & 0.099 & 0.088 \\ [.5 ex]
8 & $P(gCD)$ & 0.782 & 0.856  & 0.968 & 0.972  & 0.984  & 0.968 & 0.974  & 0.984  & 0.970 & 0.976  & 0.984   \\ [.5 ex]
   & $TPR$       &          &           & 0.495 & 0.511 & 0.547 & 0.500 & 0.515 & 0.549 & 0.499 & 0.516 & 0.553 \\ [.5 ex]
   & $FDR$       &          &           & 0.098  & 0.097 & 0.092  & 0.098  & 0.100 & 0.094 & 0.100 & 0.101 & 0.095 \\ [.5 ex]
\hline
\hline
\multicolumn{13}{|c|}{  $\lambda_1 = 0.2$, $\lambda_2 = 0.8$, $\lambda_3 = 0.5$, $\lambda_4 = 1.4$ } \\
\hline 
\hline  
 & & & & & & & & & & & & \\ [-1.5ex] 
2 & $P(gCD)$ & 0.182   & 0.186 & 0.468 & 0.464 & 0.500 & 0.468 & 0.464 & 0.502 & 0.474 & 0.466 & 0.504 \\ [.5 ex]
   & $TPR$     &             &           & 0.263 & 0.258 & 0.292  & 0.263 & 0.258 & 0.292   & 0.268 & 0.259 & 0.296    \\ [.5 ex]
   & $FDR$     &             &           & 0.099 & 0.097 & 0.098  & 0.099 & 0.097 & 0.100   & 0.101 & 0.098 & 0.098   \\ [.5 ex]
4 & $P(gCD)$ & 0.320   & 0.288 & 0.710 & 0.690 & 0.752  & 0.710 & 0.690 & 0.754  & 0.712 & 0.694 & 0.760    \\ [.5 ex]
   & $TPR$     &             &           & 0.295 & 0.289 & 0.336  & 0.295 & 0.289 & 0.337   & 0.300 & 0.292 & 0.340  \\ [.5 ex]
   & $FDR$     &             &          & 0.096 & 0.100 & 0.084   & 0.097 & 0.101 & 0.085  & 0.100 & 0.101 & 0.090  \\ [.5 ex]
8 & $P(gCD)$ & 0.738  & 0.640 & 0.908 & 0.890 & 0.934 & 0.910 & 0.890 & 0.934  & 0.912 & 0.890 & 0.940     \\ [.5 ex]
   & $TPR$     &            &           & 0.373 & 0.377 & 0.406  & 0.376 & 0.380 & 0.408  & 0.379 & 0.379 & 0.412    \\ [.5 ex]
   & $FDR$     &            &          & 0.096 & 0.098 & 0.095    & 0.096 & 0.098 & 0.097  & 0.098 & 0.099 & 0.099  \\ [.5 ex]
\hline
\end{tabular} }
\end{center}
\end{table}
  
\end{document}